\documentclass{amsart}
\usepackage{amsmath}
\usepackage{amssymb}
\usepackage{amsthm}
\usepackage{mathrsfs}
\usepackage{dsfont}
\usepackage{lineno}
\usepackage{subfigure}
\usepackage{enumerate}
\usepackage{graphicx}
\usepackage{epstopdf}
\usepackage{multirow}
\usepackage{color}
\usepackage{xspace}
\usepackage{pdfsync}
\usepackage{hyperref} 
\usepackage{algorithmicx}
\usepackage{algpseudocode}
\usepackage{algorithm}
\graphicspath{{Figures/}}

\DeclareMathOperator*{\argmin}{argmin}

\newcommand{\bq}{\begin{equation}}
\newcommand{\eq}{\end{equation}}
\newcommand{\R}{\mathbb{R}}
\newcommand{\abs}[1]{\left\vert#1\right\vert}

\newcommand{\bO}{\mathcal{O}}

\newcommand{\one}{\mathds{1}}
\newcommand{\Dt}{\mathcal{D}}

\newcommand{\MA}{Monge-Amp\`ere\xspace}

\newcommand{\diag}{\text{diag}}
\newcommand{\vp}{v^\perp}
\newcommand{\M}{\mathcal{M}}

\algnewcommand{\LineComment}[1]{\State \(\triangleright\) #1}

\newtheorem{theorem}{Theorem}
\newtheorem{thm}{Theorem}
\theoremstyle{lemma}
\newtheorem{lemma}[theorem]{Lemma}

\newtheorem{defn}{Definition}

\theoremstyle{remark}

\newtheorem{rem}{Remark}

\begin{document}
\title{OPTIMAL TRANSPORT FOR SEISMIC FULL WAVEFORM INVERSION}

\author{Bj\"orn Engquist}
\thanks{The research was supported by the Texas Consortium for Computational Seismology. The first and the third author were also partially supported by NSF grant DMS-1522792.}
\address{Department of Mathematics and ICES, The University of Texas at Austin, 1 University Station C1200, Austin, TX 78712 USA}
\email{engquist@math.utexas.edu}

\author{Brittany D. Froese}
\address{Department of Mathematical Sciences, New Jersey Institute of Technology, University Heights, Newark, NJ 07102 USA}
\email{bdfroese@njit.edu}

\author{Yunan Yang}
\address{Department of Mathematics and ICES, The University of Texas at Austin, 1 University Station C1200, Austin, TX 78712 USA}. 
\email{yunanyang@math.utexas.edu}


\begin{abstract}
Full waveform inversion is a successful procedure for determining properties of the earth from surface measurements in seismology. This inverse problem is solved by a PDE constrained optimization where unknown coefficients in a computed wavefield are adjusted to minimize the mismatch with the measured data. We propose using the Wasserstein metric, which is related to optimal transport, for measuring this mismatch. Several advantageous properties are proved with regards to convexity of the objective function and robustness with respect to noise. The Wasserstein metric is computed by solving a \MA equation. We describe an algorithm for computing its Frechet gradient for use in the optimization. Numerical examples are given.
\end{abstract}

\date{\today}
\maketitle

\section{Introduction}\label{intro}

A central step in seismic exploration is the estimation of basic geophysical properties. This can, for example, be wave velocity, which is what we will consider here. This step is typically the basis of an imaging process to determine geophysical structures.

The computational technique full waveform inversion (FWI) was introduced to seismology in~\cite{lailly1983seismic,TarantolaInversion}. This inverse method follows the common strategy of PDE constrained optimization. 
In this context, the goal is to recover the unknown wave velocity $v(x,z)$ from the resulting wavefield $u(x,z,t)$.  In practice, measurments are available only at the surface, and the velocity field needs to be recovered from the surface measurement
\[ g = u(x,0,t). \]

In two-dimensions, for example, this can be modeled by the acoustic wave equation in the time domain:
\bq\begin{cases}\label{eq:forward}
u_{tt}(x,z,t)-v(x,z)^2\left(u_{xx}(x,z,t)+u_{zz}(x,z,t)\right)=0,\\
u(x,z,0) = u_0(x,z),\\
u_t(x,z,0) = 0,
\end{cases}
\eq 
where $u_0(x,z)$ is the initial wave field generated by a Ricker wavelet signal~\cite{zhang2013full}.
For a given wavefield $v$, the solution of this wave equation yields the simulated data
\bq\label{eq:surface} f(v) = u(x,0,t). \eq

In an ideal setting, the observed data also solves this forward problem so that
\[ g = f(v^*) \]
where $v^*$ is the true velocity field.  However, this is unlikely in practice due to the presence of noise, measurement errors, and modeling errors.  Instead, the goal of full waveform inversion is to estimate the the true velocity field through the solution of the optimization problem
\bq\label{eq:fwi}
\tilde{v} = \argmin\limits_v \,d(f(v),g)
\eq
where $d(f,g)$ is some measure of the misfit between two signals.

Two primary concerns in full waveform inversion are the well-posedness of the underlying model recovery problem and the suitability of the misfit $d(f,g)$ for minimization.
In this work we will focus on properties of the misfit measure $d$ and not on the overall question of uniqueness and stability of the inverse problem.  For additional details on uniqueness and stability, we refer to~\cite{isakov2006inverse,stefanov2013recovery,sylvester1987global}.

The $L_2$ norm is often used to measure the misfit, which typically generates many local minima and is thus unsuitable for minimization.  This problem is exacerbated by the fact that measured signals usually suffer from noise in the measurements~\cite{masoni2013alternative}.  

In~\cite{EFWass}, we proposed using the Wasserstein metric for the misfit function, i.e. $d(f,g)=W_2^2(f,g)$. The Wasserstein metric measures the distance between two distributions as the optimal cost of rearranging one distribution into the other~\cite[p.~207]{Villani}.  
The mathematical definition of the distance between the distributions $f:X\to\R^+$, $g:Y\to\R^+$ can be formulated as
\bq\label{eq:W2}  W_2^2(f,g) = \inf\limits_{T\in\M} \int\limits_X\abs{x-T(x)}^2f(x)\,dx \eq
where $\M$ is the set of all maps that rearrange the distribution $f$ into $g$. 

We cannot directly compute the Wasserstein metric between two wave fields since these are not probability measures.  Some additional processing is needed to ensure that the signals are strictly positive and have unit mass.  This can be easily done by separately comparing the positive and negative parts of the signals, which are then rescaled to have mass one.  We define $f^+ = \max\{f,0\}$, $f^-=\max\{-f,0\}$, and $\langle f\rangle = \int_X f(x)\,dx$.  With this notation, we propose solving the optimization problem~\eqref{eq:fwi} using the misfit
\bq\label{eq:misfit}
d(f,g) = W_2^2\left(\frac{f^+}{\langle f^+\rangle},\frac{g^+}{\langle g^+\rangle}\right) + W_2^2\left(\frac{f^-}{\langle f^-\rangle},\frac{g^-}{\langle g^-\rangle}\right).
\eq

It is our goal to prove several desirable properties relating to convexity and insensitivity to noise, which were briefly discussed in~\cite{EFWass}. Another important contribution in this paper is a derivation of the gradient of $d(f(v),g)$ with respect to $v$, which is essential for gradient based minimization algorithms. It is outside the scope of this work to study serious applications, but we  give some numerical examples to show the quantitative behavior and to compare with the simple search algorithm used in~\cite{EFWass}. In this earlier paper, simple geometrical optics was used in the forward problem. Here we consider the full wave equation.

We briefly recall one example from~\cite{EFWass} that illustrates the advantage of the Wasserstein metric. Consider the misfit between the simple wavelet $f$ in Figure~\ref{fig:wavelet_prof} and another wavelet shifted by a distance $s$.  Figures~\ref{fig:wavelet1d_L2}-\ref{fig:wavelet1d_W2_pm} illustrate that the $L_2$ norm is constant when $s$ is large and has many local minima. On the other hand, the Wasserstein metric is uniformly convex with respect to shifts, which are natural in travel time mismatches. 

\begin{figure}[htp]
\begin{center}
\subfigure[]{\includegraphics[width=0.49\textwidth]{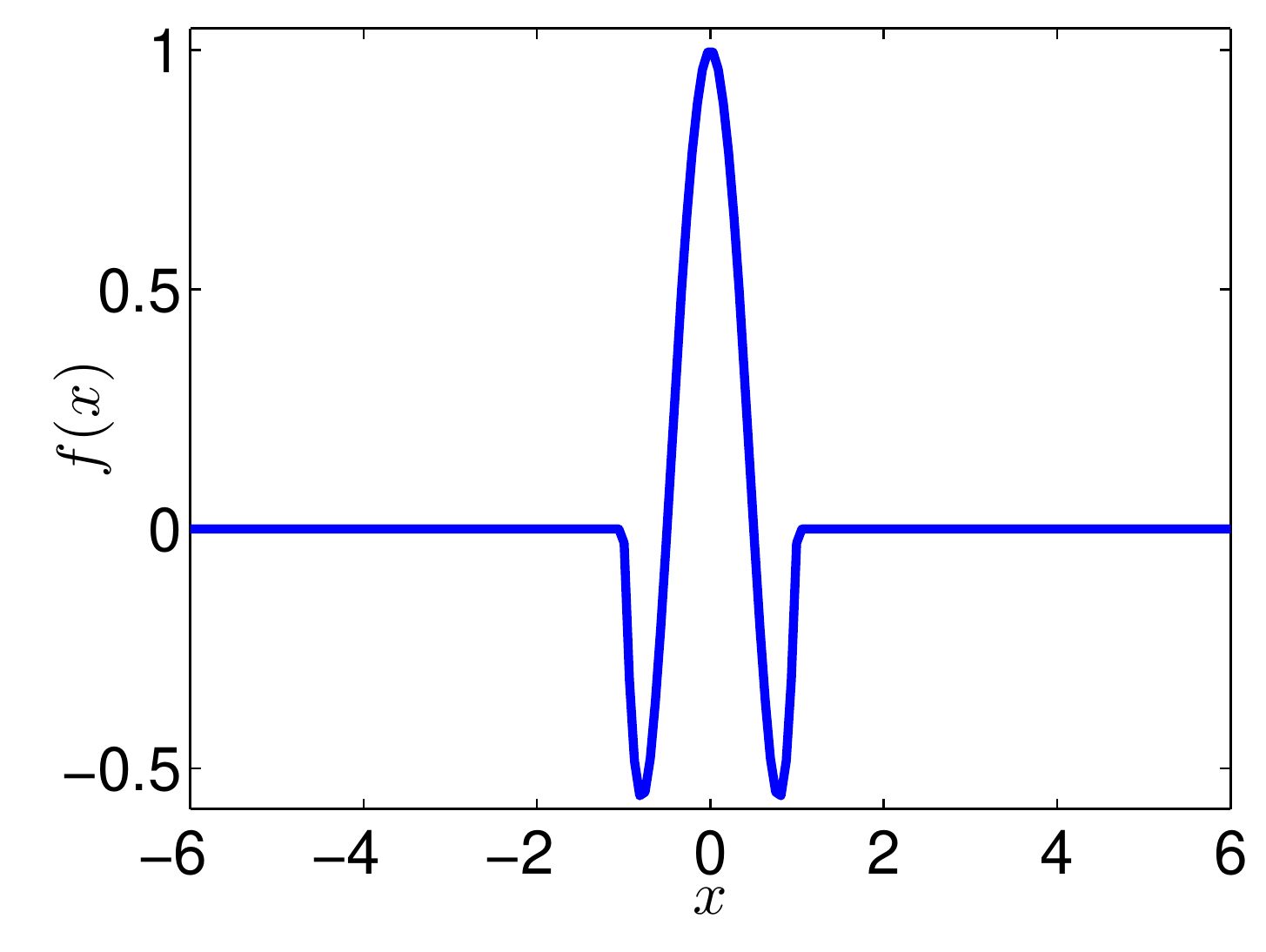}\label{fig:wavelet_prof}}\\
\subfigure[]{\includegraphics[width=0.49\textwidth]{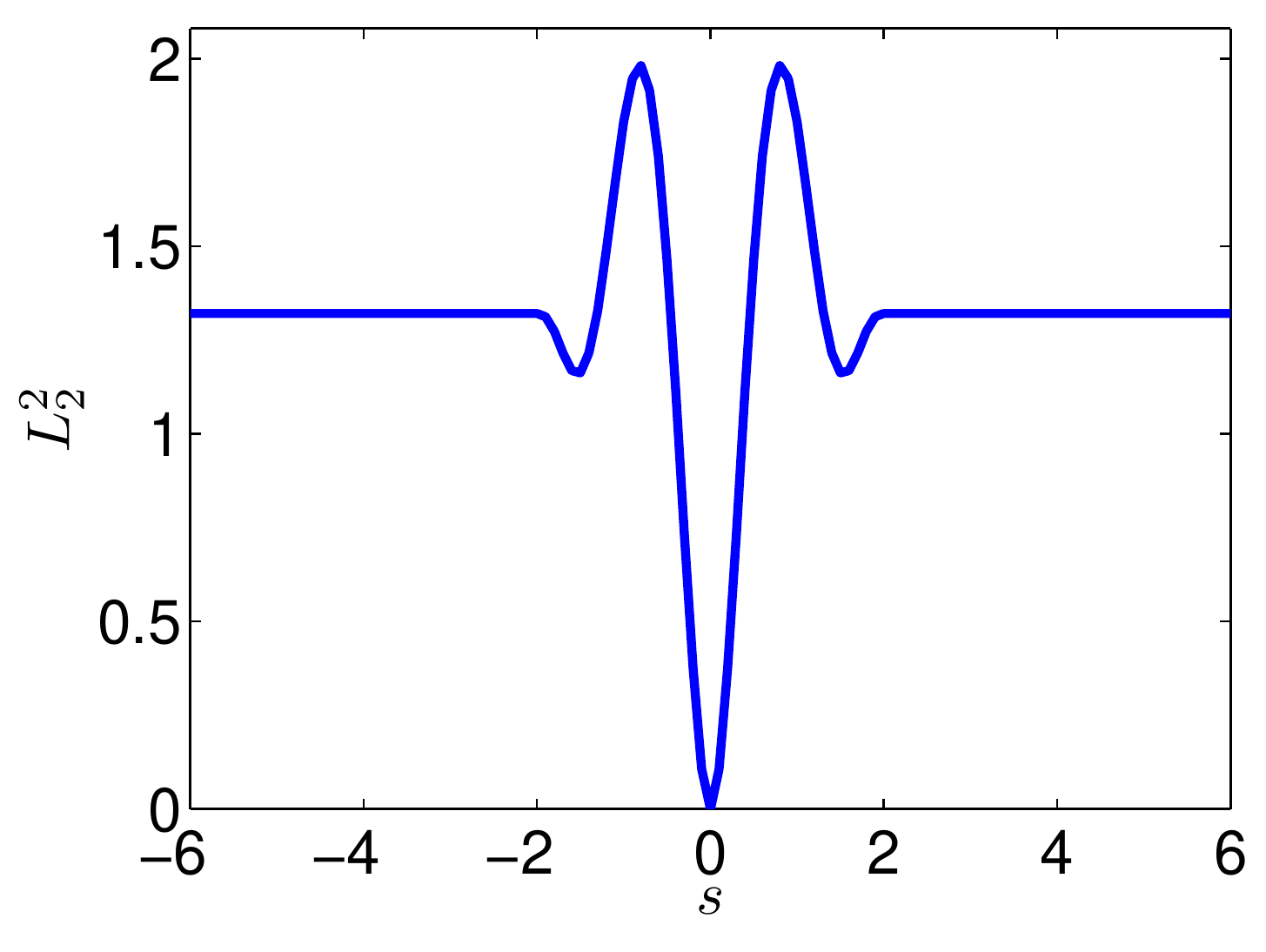}\label{fig:wavelet1d_L2}}
\subfigure[]{\includegraphics[width=0.49\textwidth]{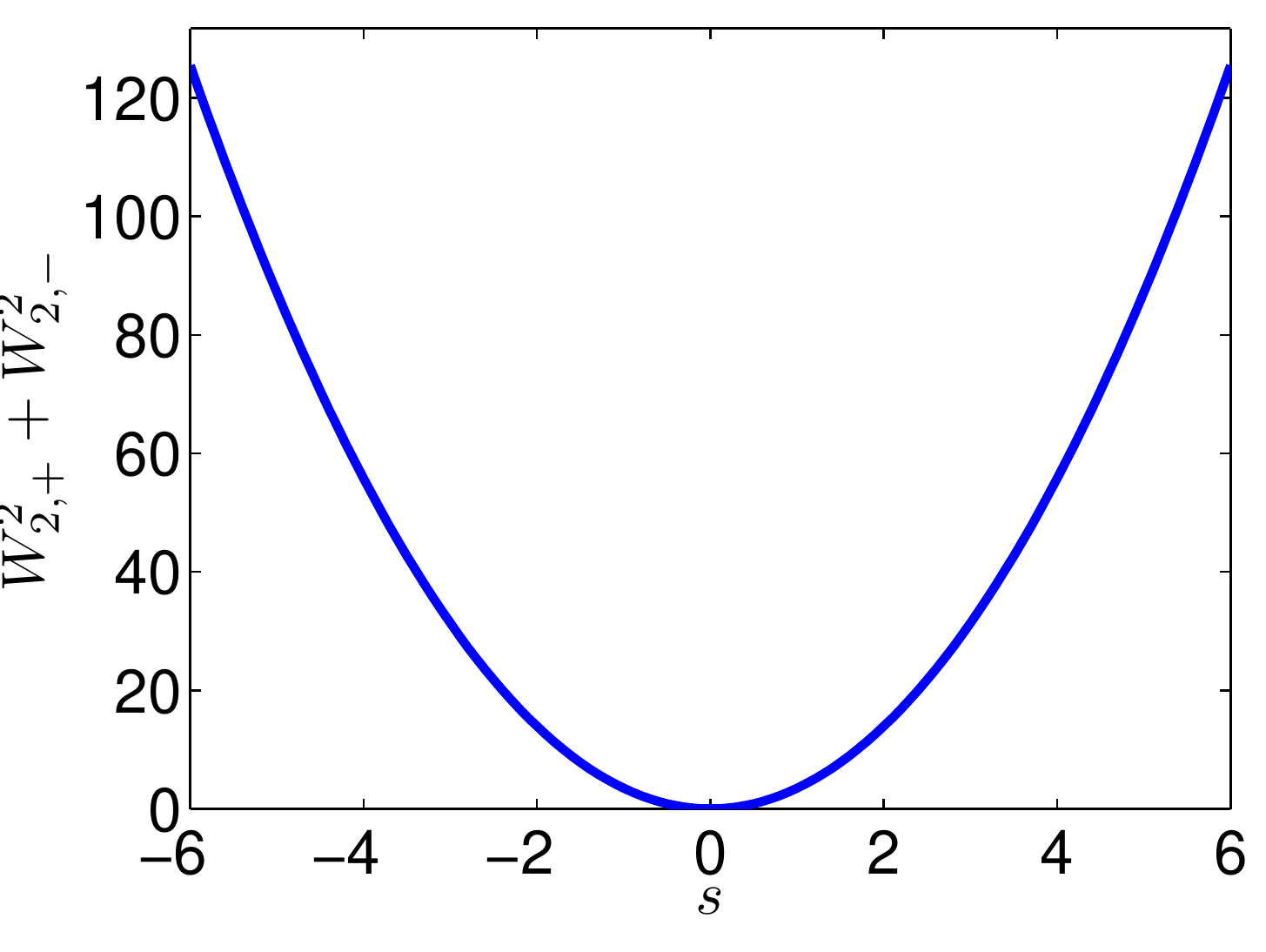}\label{fig:wavelet1d_W2_pm}}
\caption{\subref{fig:wavelet_prof}~A wavelet profile $f(x)$.  The distance between $f(x)$ and $g(x) = f(x-s)$ measured by \subref{fig:wavelet1d_L2}~$L_2^2(f,g)$, and \subref{fig:wavelet1d_W2_pm}~$W_2^2(f^+,g^+) + W_2^2(f^-,g^-)$\cite{EFWass}. }
\label{fig:wavelet}
\end{center}
\end{figure}

Earlier algorithms for the numerical computation of the Wasserstein metric required a large number of operations~\cite{BenBre,Bertsekas,Bosc}.  The optimal transportation problem can be rigorously related to the following \MA equation~\cite{Brenier, KnottSmith}, which enables the construction of more efficient methods for computing the Wasserstein metric.
\bq\label{eq:MA} \begin{cases}
\det(D^2u(x)) = f(x)/g(\nabla u(x)) + \langle u \rangle, & x\in X\\
u \text{ is convex.}
\end{cases}\eq
The Wasserstein metric is then given by
\bq\label{eq:W2_u}
W_2^2(f,g) = \int\limits_X \abs{x-\nabla u(x)}^2f(x)\,dx.
\eq

There are now fast and robust numerical algorithms for the solution of~\eqref{eq:MA}, and thus for the computation of $W_2^2$, and these form the basis for our numerical techniques~\cite{benamou2014numerical}.

In~\autoref{sec:convexity}, we will study the convexity of the quadratic Wasserstein metric with respect to shift, dilation, and partial amplitude change. Errors between simulated and observed data in the form of shifts and dilations can  occur naturally from an incorrect velocity model, while inaccurate measurements or variations in the strength of reflecting surfaces can result in larger or smaller local amplitudes. We will give a rigorous proof of these convexity statements using the fundamental theorem of optimal transport and convexity of the Monge-Kantorovich minimization problem\cite[p.~80]{Villani}.

In~\autoref{sec:noise}, we will discuss how the Wasserstein metric is affected by random noise with a uniform distribution. For the optimal transport problem on the real line, both a theorem and numerical illustrations will be given to show that the effect of noise is negligible. For higher dimensions, we estimate the effect of noise by finding an upper bound.

We review an efficient numerical method for computing the Wasserstein metric via the numerical solution of the nonlinear elliptic \MA partial differential equation in~\autoref{sec:numerics}. After obtaining the discrete solution, we can easily approximate the squared Wasserstein metric. 

We are interested in recovering the parameters in the forward wave equation by minimizing the Wasserstein metric between simulated and observed data. In~\autoref{sec:linear}, we describe a method for numerically obtaining the gradient of the Wasserstein metric by first discretizing the metric, then linearizing the result. This approach is particularly straightforward for the numerical method utilized in this paper.

Finally, numerical examples presented in~\autoref{sec:results} show the quantitative and qualitative behavior of the minimization procedure. Parameters in low dimensional model problems are recovered by minimizing the Wasserstein metric without any need to regularize the problem.

          \section{Convexity of the quadratic Wasserstein metric}\label{sec:convexity}
In most optimization problems, convexity of the objective function is a desirable property. The example of convexity given in Figure~\ref{fig:wavelet1d_W2_pm} was our motivation for considering the Wasserstein metric in the context of full waveform inversion. In this section, we will mathematically study this convexity with respect to variations that are common in the context of seismic exploration.  In particular,
we analyze cases where $f$ is derived from $g$ by either a local change of amplitude or a linear change of variables in the form of a shift or dilation.

The shift and dilation are 
typical effects of variations in the velocity $v$, as can be seen in a simple example.
A one-dimensional, constant velocity model is
\begin{equation*}
\begin{cases}
\frac{\partial^2 u}{\partial t^2} = v^2 \frac{\partial^2u}{\partial x^2},\quad &x>0, t>0, \\
u = 0, \quad\frac{\partial u}{\partial t} = 0,\quad &x>0, t=0,\\
u = u_0(t),\quad &x=0,t>0.
\end{cases}
\end{equation*}
One solution to the equation is $u(x,t;v) = u_0(t-x/v)$. For fixed $x$, variations in $v$ induce shifts in the signal. When $t$ is fixed, variation of $v$ generates dilation in $u_0$ as a function of $x$.

Local changes in amplitude can originate from variations in the strength of reflecting surfaces.  A material composed of layers of different materials will yield a velocity field that is (approximately) piecewise constant.  The strength of the reflection at a discontinuity is, in turn, related to the velocity field in each layer.  An incorrect estimation of the velocity field will then lead to larger or smaller local amplitudes in the resulting wavefield. 

We note that any shifts, dilations, and amplitude changes in a signal~$f$ will correspond to a similar transformation in the positive and negative parts of~$f$.  Thus in the results below, it is sufficient to assume that the original profile~$f$ is non-negative.  The results below will still consider any changes in mass that result from the transformations.

\subsection{Convexity with respect to shift}
We begin by assuming looking at the effects of shifting a density function~$f$, which has no effect on the total mass.

\begin{thm}[Convexity of shift]
\label{thm:shift}
Suppose $f$ and $g$ are probability density functions of bounded second moment. Let $T$ be the optimal map that rearranges $f$ into $g$. If $f_s(x)=f(x-s\eta)$ for $\eta \in \R^n$, then the optimal map from $f_s(x)$ to $g(y)$ is $T_s= T(x - s\eta)$. Moreover, $W_2^2(f_s,g)$ is convex with respect to the shift size $s$.
\end{thm}

The proof relies on the concept of cyclical monotonicity, which can be used to characterize an optimal map.

\begin{defn}[Cyclical monotonicity]
\label{cyclical}
We say that a map $T:X\subset\R^n\to\R^n$ is cyclically monotone if for any $m\in\mathbb{N}^+$, $x_i \in X$, $1\leq i \leq m$, $x_0 \equiv x_m$,
\bq\label{eq:cyclical}
\sum_{i=1}^{m}x_i \cdot T(x_i) \geq  \sum_{i=1}^{m} x_i \cdot T(x_{i-1}).
\eq 
\end{defn}

\begin{thm}[Optimality criterion for quadratic cost~{\cite[Theorem~2.13]{ambrosio2013user}}]\label{thm:optimality}
Let $f$ and $g$ be probability density functions supported on sets $X$ and $Y$ respectively.  A mass-preserving map $T:X\to Y$ minimizes the quadratic optimal transport cost if and only if it is cyclically monotone.
\end{thm}

\begin{proof}[Proof of Theorem~\ref{thm:shift}]
By construction, the map $T_s$ rearranges $f_s$ into $g$.  The cyclical monotoncity of $T_s$ follows immediately from the cyclical monotoncity of $T$:
\begin{align*}
\sum_{i=1}^{m}x_i \cdot T_s(x_i) &= \sum_{i=1}^{m}[(x_i-s\eta)+s\eta]\cdot T(x_i-s\eta)\\
  &\geq \sum_{i=1}^{m}[(x_i-s\eta)+s\eta]\cdot T(x_{i-1}-s\eta)\\ 
	&= \sum_{i=1}^{m}x_i \cdot T_s(x_{i-1}).
\end{align*}
     
Then the squared Wasserstein metric can be expressed as
     \begin{align}
     W_2^2(f_s,g) & =  \int\left|x-T_s(x)\right|^2f_s(x)dx\nonumber =  \int\left|x-T(x-s\eta)\right|^2f(x-s\eta)dx \nonumber\\
     & =  W_2^2(f,g)+s^2|\eta|^2+2s\int\eta\cdot(x-T(x))f(x)\,dx.
     \end{align}\label{eqn:convex}
     The convexity with respect to $s$ is evident from the last equation.
\end{proof}

 \subsection{Convexity with respect to dilation}
Next we consider the convexity of the Wasserstein metric with respect to dilations or contractions of the density functions.  We begin by characterizing the optimal map in this setting.

\begin{lemma}[Optimal map for dilation]\label{lem:dilation}
Assume $f$ and $g$ are probability density functions of bounded second moment satisfying $f(x)=\det(A)g(Ax)$, where $A$ is a symmetric positive definite matrix. Then the optimal transport map rearranging $f(x)$ into $g(y)$ is $T(x)=Ax$.
\end{lemma}

\begin{proof}
    Again, the cyclical monotonicity condition of Theorem~\ref{thm:optimality} is the key to verifying optimality.  
		
Since $A$ is symmetric positive definite, it has a unique Cholesky decomposition  $A=L^TL$ for some upper triangular matrix $L$. Then for any $x_i \in X$,
\begin{align*}
\sum_{i=1}^m x_i \cdot (T(x_i)-T(x_{i-1})) &=  \frac{1}{2}\sum_{i=1}^m (x_{i-1}^TL^TLx_{i-1} + x_i^TL^TLx_i - 2x_i^TL^TLx_{i-1})\\
	&= \frac{1}{2}\sum_{i=1}^m \abs{Lx_i - Lx_{i-1}}^2\\
	&\geq 0,
\end{align*}
which verifies the optimality condition.
\end{proof}

\begin{rem}
The requirement that $A$ be symmetric positive definite is necessary for $y = Ax$ to be the optimal map. For example, let $A$ be the rotation matrix 
\(\left(
\begin{array}{ccc}
\cos\theta & \sin\theta \\
-\sin\theta & \cos\theta
\end{array}
\right)\) with $\theta = \pi$ and let $g$ satisfy the symmetry condition
$g(x,y) = g(-x,-y)$.
 Then the optimal map from $f(x) = g(Ax)$ to $g$ is the identity function instead of $T(x) = Ax$.
\end{rem}

Convexity is a separate issue as it depends on the parameterization. A special case of dilation occurs when $A$ is a diagonal matrix. The following theorem is a generalization where the dilation need not occur along coordinate directions.

\begin{thm}[Convexity with respect to dilation]\label{thm: dilconvex}
Assume $f(x)$ is a probability density function and  $g(y) = f(A^{-1}y)$ where $A$ is a symmetric positive definite matrix. Then the squared Wasserstein metric $W_2^2(f,g/\langle g\rangle)$ is convex with respect to the eigenvalues $\lambda_1,\dots,\lambda_n$ of $A$.
\end{thm}

\begin{proof}
In order to define the Wasserstein metric, it is necessary to work with the normalized density $g / \langle g \rangle = \det(A)^{-1}f(A^{-1}y)$. 
By Lemma~\ref{lem:dilation}, the optimal mapping is
\[ T(x) = Ax = O\Lambda O^T \]
where $O$ is an orthogonal matrix and $\Lambda$ is a diagonal matrix whose entries are the eigenvalues $\lambda_1, \ldots, \lambda_n$.  Then the squared Wasserstein metric can be expressed as
\begin{align*}
W_2^2\left(f,\frac{g}{\langle g \rangle}\right) &= \int f(x)\abs{x-Ax}^2\,dx\\
  &= \int f(x) x^TO (I-\Lambda)^2O^Tx\,dx\\
	&= \int f(Oz) z^T (I-\Lambda)^2z\,dz,
\end{align*}
which is convex in $\lambda_1, \ldots, \lambda_n$.
\end{proof}

    \begin{rem}
    If both dilation and shift are present, the Wasserstein metric will be convex with respect to each of the corresponding parameters. 
    \end{rem}

    \subsection{Convexity with respect to partial amplitude change}
Finally, we consider the problem where a profile $f$ is derived from $g$, but with a decreased amplitude in part of the domain.  That is, we suppose that the domain is decomposed into $\Omega = \Omega_1 \cup \Omega_2$ with $\Omega_1\cap\Omega_2 = \emptyset$.  For an amplitude loss parameter $0 \leq \beta \leq 1$ we suppose that $f$ depends on the probability density function~$g$ via
\bq
f_\beta(x) = \begin{cases}
\beta g(x), & x\in\Omega_1\\
g(x), & x\in\Omega_2.
\end{cases}
\eq

\begin{thm}[Convexity with respect to partial amplitude loss]\label{thm:ampLoss}
The squared Wasserstein metric~$W_2^2(f_\beta/\langle f_\beta\rangle, g)$ is a convex function of the parameter $\beta$.
\end{thm}

In order to prove this result, we introduce an alternative form of the rescaled density in terms of a parameter $-1 \leq \alpha \leq 0$.
    \[
    h_\alpha(x) =
		\begin{cases}(1+\alpha)g(x), & x\in\Omega_1,\\
                (1-\gamma_\alpha)g(x), & x\in\Omega_2,
		\end{cases}
		\]
		where
		\[
    \gamma_\alpha = \alpha  \frac{\int_{\Omega_1} g}{\int_{\Omega_2}g}.
		\]
Note that $h_\alpha$ is non-negative and has unit mass by construction, with $h_0 = g$.

The two different families of parameters are connected as follows.
\begin{lemma}[Parameterization of amplitude loss]\label{lem:parameters}
Define the parameterization function
\[ \alpha(\beta) = \frac{\beta}{\beta\int_{\Omega_1}g + \int_{\Omega_2}g} - 1. \]
Then $\alpha:[0,1]\to[-1,0]$ is concave and the associated density functions are related through
\[ \hat{f}_\beta \equiv \frac{f_\beta}{\langle f_\beta \rangle} = h_{\alpha(\beta)}.\]
\end{lemma}

The proof of convexity with respect to $\beta$ will come via convexity with respect to $\alpha$.
\begin{lemma}[Convexity with respect to partial amplitude change]\label{lem:amplitude}
The squared Wasserstein metric $W_2^2(h_\alpha,g)$ is a convex function of the parameter $\alpha$. 
\end{lemma}

\begin{proof}
Choose any $\alpha_1,\alpha_2 \in [-1,0]$ and $s\in [0,1]$. From convexity of the Monge-Kantorovich minimization problem~\cite[p.~220]{Villani}, we have
    \begin{equation}\label{ineq}
    W_2^2 
		\left(sh_{\alpha_1}+(1-s)h_{\alpha_2}, g\right) \leq s W_2^2\left(h_{\alpha_1},g\right)+\left(1-s\right)W_2^2(h_{\alpha_2},g).
    \end{equation}

We can calculate
    \begin{align*}
    sh_{\alpha_1}+(1-s)h_{\alpha_2} &=
    \begin{cases}
    s(1+\alpha_1)g+(1-s)(1+{\alpha_2})g, & x\in\Omega_1,\\
    s(1-\gamma_{\alpha_1})g+(1-s)(1-\gamma_{{\alpha_2}})g, & x\in\Omega_2.
    \end{cases}
		\\
		&=
    \begin{cases}
    (1+s\alpha_1+{\alpha_2}-s{\alpha_2})g, & x\in\Omega_1,\\
    (1-\gamma_{s\alpha_1+(1-s){\alpha_2}})g, & x\in\Omega_2,
    \end{cases}
		\\
		&= h_{s\alpha_1+(1-s){\alpha_2}}.
    \end{align*}

    Thus we can rewrite Equation~\eqref{ineq} as
    \bq 
    W_2^2(h_{s\alpha_1 + (1-s){\alpha_2}},g) \leq sW_2^2(h_{\alpha_1})+(1-s)W_2^2(h,f_{\alpha_2},g)
    \eq
 and the Wasserstein metric $W_2^2(h_\alpha,g)$ is convex with respect to $\alpha$.
\end{proof}          

A simple consequence of this result is that the misfit is a decreasing function of~$\alpha$.
\begin{lemma}[Misfit is non-increasing]\label{lem:decreasing}
Let $-1 \leq \alpha_1 < \alpha_2 \leq 0$.  Then $W_2^2(h_{\alpha_2}, g) \leq W_2^2(h_{\alpha_1},g)$.
\end{lemma}

\begin{proof}
Define the parameter $s\in[0,1]$ by $s = \alpha_2 / \alpha_1$.  Then we can use the convexity result of Lemma~\ref{lem:amplitude} to compute
\begin{align*}
W_2^2(h_{\alpha_2},g) &= W_2^2(h_{s\alpha_1 + (1-s)\cdot 0},g)\\
  &\leq s W_2^2(h_{\alpha_1}, g) + (1-s) W_2^2(h_0, g)\\
	&\leq W_2^2(h_{\alpha_1}, g),
\end{align*}
where we have used the fact that $h_0 = g$.
\end{proof}

Using these lemmas, we can now establish convexity with respect to the natural amplitude loss parameter~$\beta$.

\begin{proof}[Proof of Theorem~\ref{thm:ampLoss}]
Choose any $\beta_1, \beta_2, s \in [0,1]$.  
From the concavity of $\alpha(\beta)$ we have
\[ \alpha(s\beta_1 + (1-s)\beta_2) \geq s\alpha(\beta_1) + (1-s)\alpha(\beta_2). \]
Applying Lemmas~\ref{lem:amplitude}-\ref{lem:decreasing} we can compute 
\begin{align*}
W_2^2(\hat{f}_{s\beta_1 + (1-s)\beta_2},g) &= W_2^2(h_{\alpha(s\beta_1 + (1-s)\beta_2)},g)\\
  &\leq W_2^2(h_{s\alpha(\beta_1) + (1-s)\alpha(\beta_2)},g)\\
	&\leq s W_2^2(h_{\alpha(\beta_1)},g) + (1-s)W_2^2(h_{\alpha(\beta_2)},g)\\
	&= s W_2^2(\hat{f}_{\beta_1},g) + (1-s)W_2^2(\hat{f}_{\beta_2},g),
\end{align*}
which establishes the convexity.
\end{proof}

\section{Insensitivity with respect to noise}\label{sec:noise}

In the practical application of full waveform inversion, it is natural to experience noise in the measured signal, and therefore robustness with respect to noise is a desirable property in a misfit function. We will show that Wasserstein metric is substantially less sensitive to noise than the $L_2$ norm.

The Wasserstein metric depends on the square of the translate $T$. This implies that if $f$ is an oscillatory perturbation of $g$ then the Wasserstein metric $W_2^2(f,g)$ is small. A simple one-dimensional example given by Villani~\cite[Exercise~7.11]{Villani} shows that $W_2^2(f_\epsilon, g) = O(\epsilon^2)$ for $f_\epsilon =\left(1+\sin\frac{2\pi x}{\epsilon}\right)$ and $g=1$ on $[0,1]$. A numerical example without analysis was given in~\cite{EFWass}.

\subsection{One dimension}
In one dimension, it is possible to exactly solve the optimal transportation problem in terms of the cumulative distribution functions
\[ F(x) = \int_0^x f(t)\,dt, \quad G(x) = \int_0^x g(t)\,dt. \]
See Figure~\ref{fig:real line}. Then it is well known~\cite[Theorem~2.18]{Villani} that the optimal transportation cost is
\bq\label{eq:cost1D}
W_2^2(f,g) = \int_0^1|F^{-1}(t)-G^{-1}(t)|^2\,dt.
\eq
If additionally the target density~$g$ is positive, then the optimal map from $f$ to $g$ is given by
\[ T(x) = G^{-1}(F(x)). \]

\begin{figure}[htp]
\begin{center}
\subfigure[]{\includegraphics[width = 0.49\textwidth]{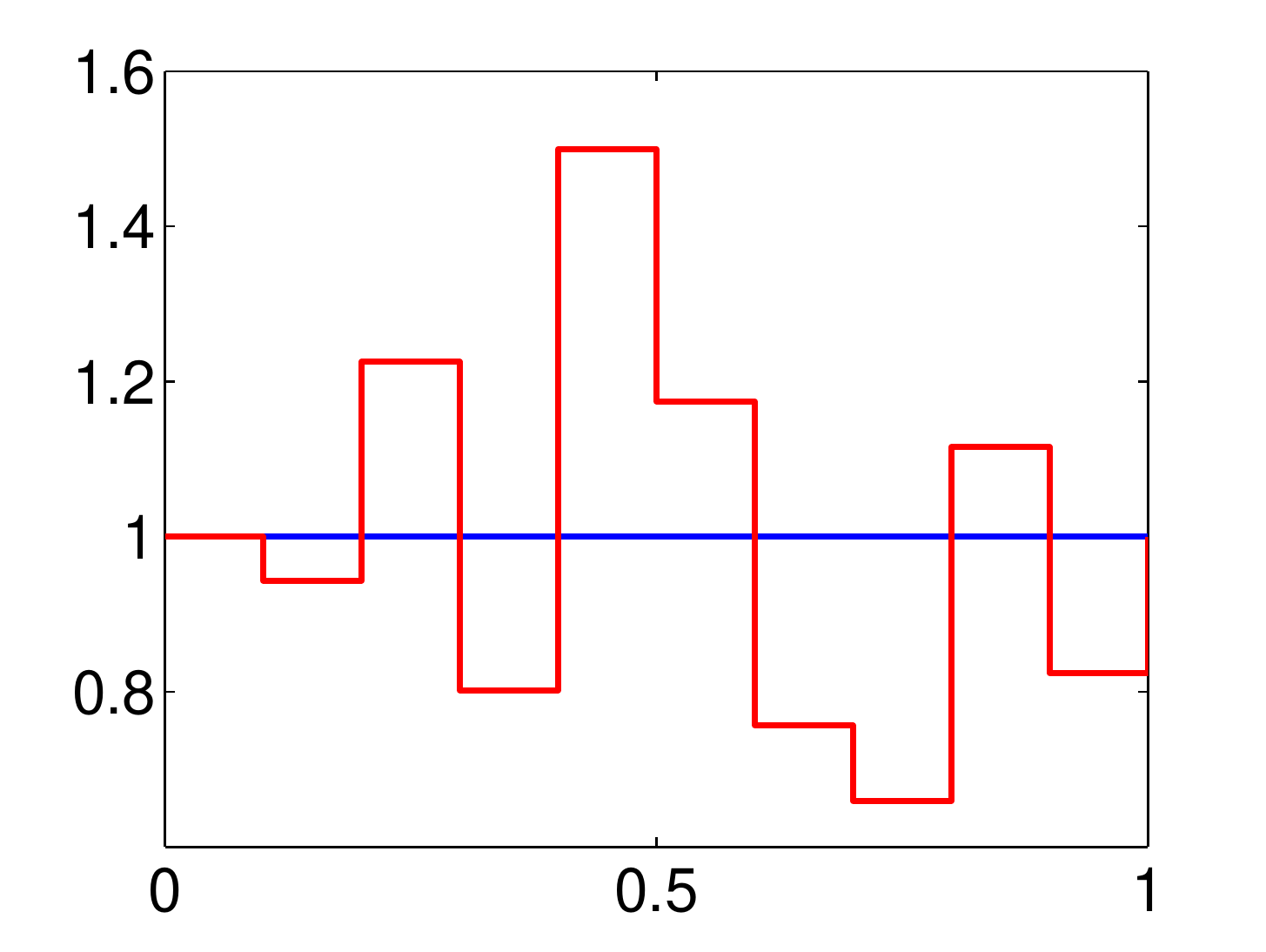}\label{fig:noise10}}
\subfigure[]{\includegraphics[width = 0.49\textwidth]{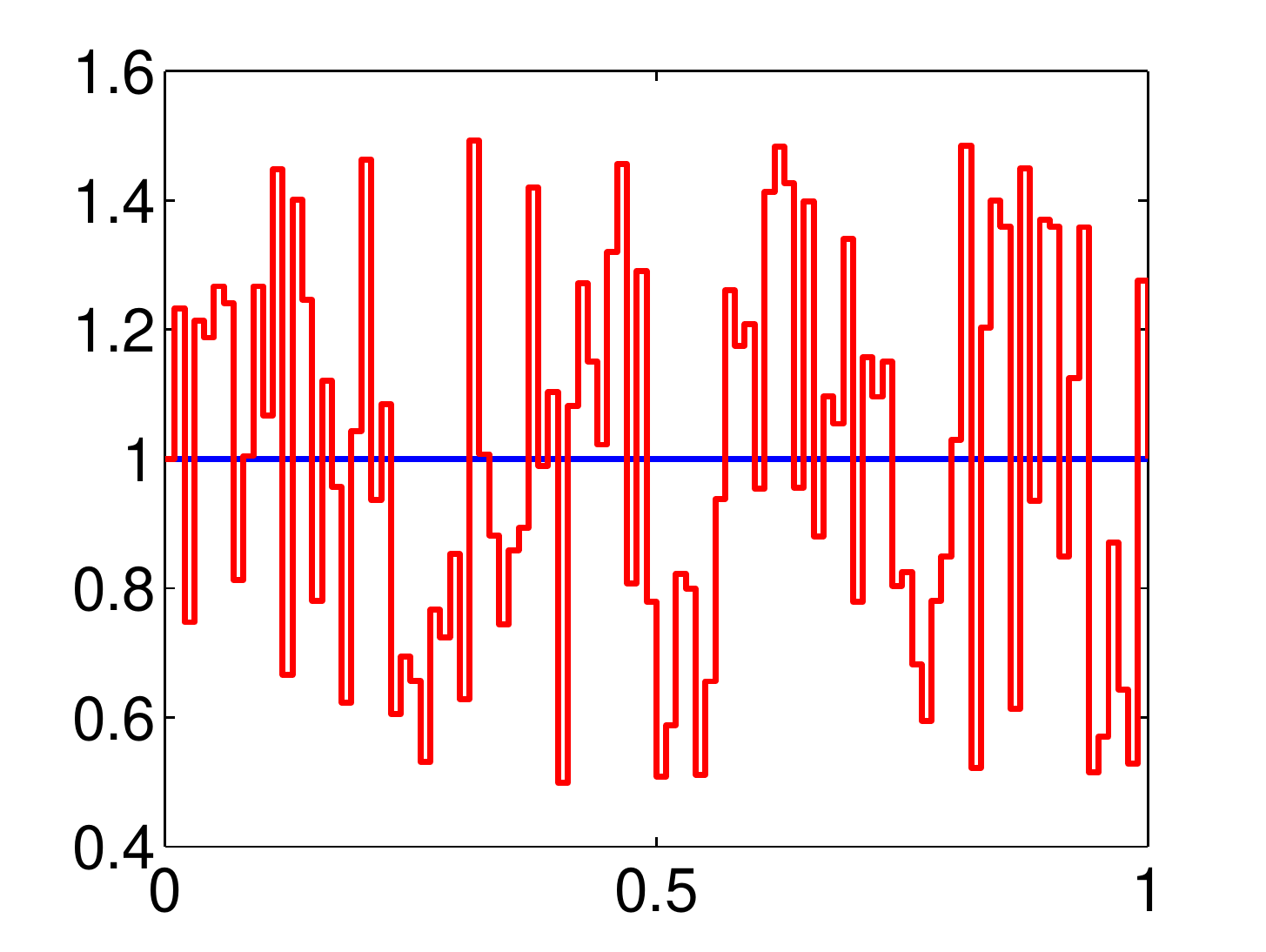}\label{fig:noise100}}
\end{center}
\vspace*{-4pt}
\caption{
Densities $g$ (blue) and $f_N$ (red, $c=-0.5$) for
\subref{fig:noise10}~$N = 10$ and \subref{fig:noise100}~$N = 100$.}
\label{fig:noise density}
\end{figure}

\begin{thm}[Insensitivity to noise in 1-D]\label{1D_noise}
Let $g$ be a positive probability density function on $[0,1]$ and choose $0<c<\min g$.  Let $f_N(x) = g(x) + r^N(x)$,
which contains piecewise constant additive noise 
$r_N$ drawn from the uniform distribution $U[-c,c]$.
Then $\mathbb{E}W_2^2(f_N/\langle f_N \rangle, g) = \bO(\frac{1}{N})$.
\end{thm}

\begin{figure}
\centering
\subfigure[]{\includegraphics[width = 0.49\textwidth]{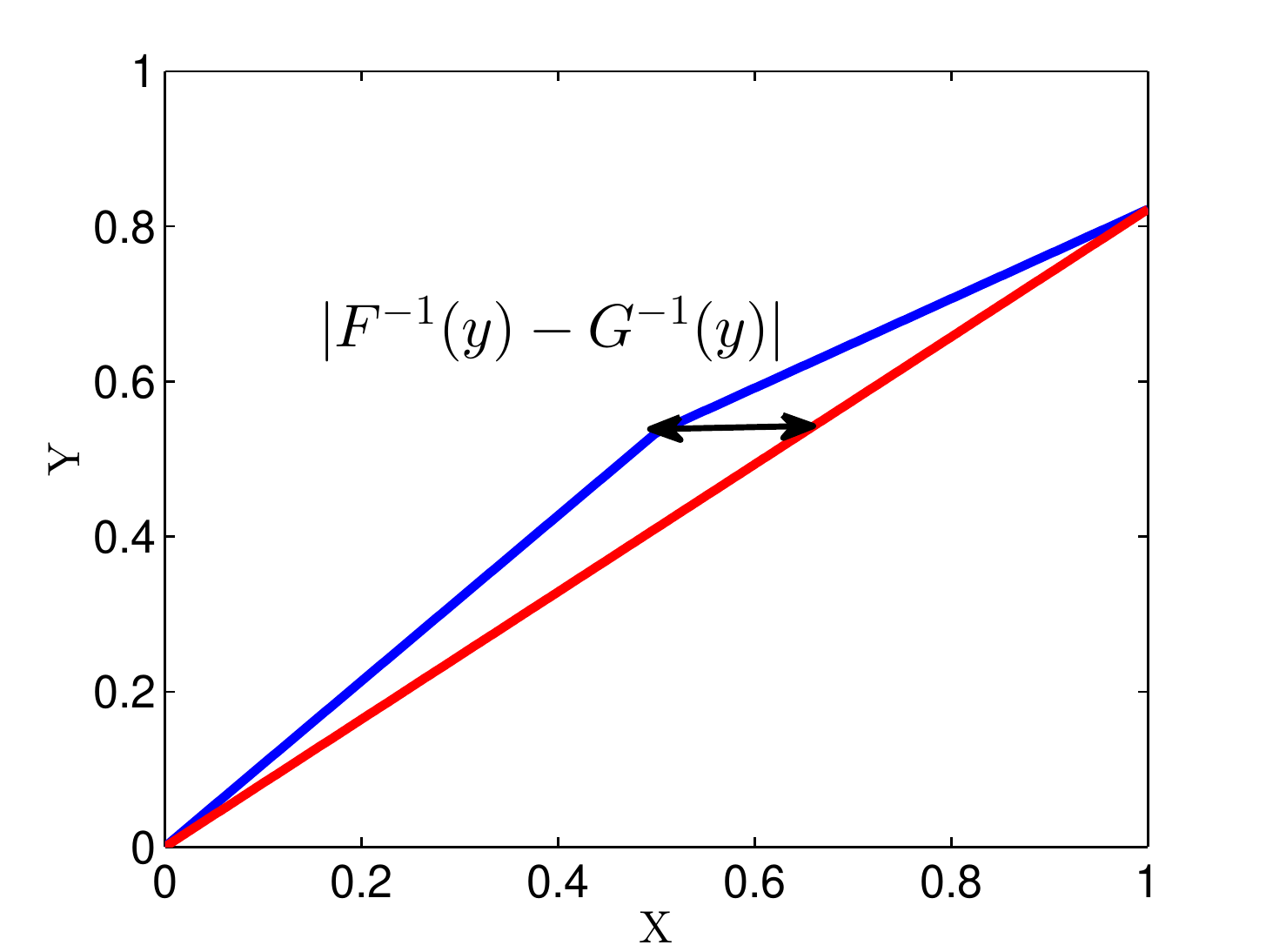}\label{fig:cdf}}
\subfigure[]{\includegraphics[width = 0.49\textwidth]{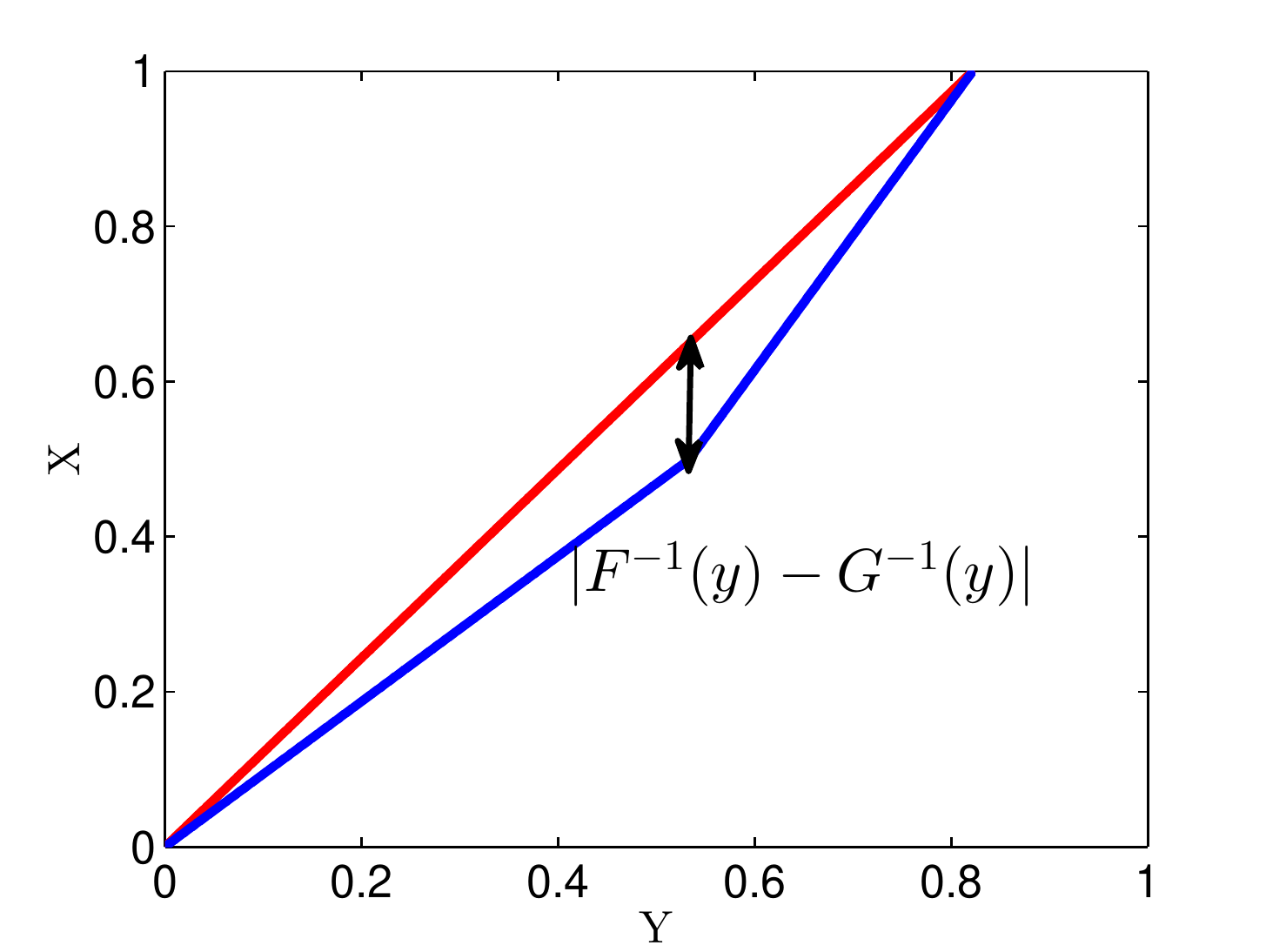}\label{fig:invcdf}}
\caption{\subref{fig:cdf}~Cumulative distribution functions $F(x)$ (red),  $G(x)$ (blue) and \subref{fig:invcdf}~the inverse  functions~$F^{-1}(y)$ (red), $G^{-1}(y)$ (blue).}
\label{fig:real line}
\end{figure}
Without loss of generality, we take $g=1$ on $[0,1]$. Figure~\ref{fig:noise density} shows the effect of the noise. For $x\in \left(\frac{i-1}{N},\frac{i}{N}\right]$, \( r^N(x) \equiv r_i\), with each $r_i$ drawn from $U[-c,c]$. As $N\rightarrow \infty$, $r^N(x)$ approximates the noise function $r(x)$ on $[0,1]$. For any $x_0 \in [0,1]$, $r(x_0)$ is a random variable with uniform distribution $U[-c,c]$.

\begin{proof}[Proof of Theorem~\ref{1D_noise}]
For each $i$, $r_i$ is a random variable of uniform distribution $U[-c,c]$, $0<c<\min g$. Thus, we have $\mathbb{E}r_i = 0$ and $\mathbb{E}\bar{r} = 0$. 

Let $h = 1/N$ and $x_i = ih$ for $i = 0,\ldots,N$.  Then the noisy density function is given by
\[ f_N(x) = 1+r_i, \quad x\in (x_{i-1},x_i]. \]
We begin by calculating the Wasserstein metric between $f_N$ and the constant $g_N = 1 + \bar{r}^N$, which share the same mass.

In order to make use of~\eqref{eq:cost1D}, we derive the cumulative distribution function and its inverse for both $f_N$ and $g_N$:
\begin{align*}
F_N(x) &= \sum\limits_{j=1}^{i-1} (1+r_j)h + (1+r_i)(x-x_{i-1}), & x \in (x_{i-1},x_i]\\
G_N(x) &= (1+\bar{r}^N)x, & x\in[0,1] \\
F_N^{-1}(x) &= \frac{x + \left((i-1)r_i - \sum_{j=1}^{i-1} r_j\right)h}{1+r_i}, & x \in \left(\sum_{j=1}^{i-1}(1+r_j)h, \sum_{j=1}^{i}(1+r_j)h\right] \\
G_N^{-1}(x) &= \frac{x}{1+\bar{r}^N}, & x \in [0,1+\bar{r}^N],
\end{align*}
where $1 \leq i \leq N$.
%
%
%
%
%
%
%

Then we can bound the squared Wasserstein metric by
\bq
W_2^2( f_N, g_N) = \int_0^{1+\bar{r}}|F_N^{-1}(t)-G_N^{-1}(t)|^2dt
\leq  \frac{2h^3}{(1-c)^2}\sum_{i=1}^N\left(\sum_{l=1}^i r_l - ih\sum_{k=1}^N r_k\right)^2.
\eq
Since the noise $\{r_i\}_{i=1}^N$ is i.i.d., we obtain the following upper bound for the expectation of the Wasserstein metric: 
\[
\mathbb{E} W_2^2(f_N,g_N) \leq C\cdot h^3 \cdot \sum_{i=1}^N i \cdot \mathbb{E}r_1^2 \leq \frac{C_2}{N}.
\]

We can similarly establish a lower bound so that
\begin{equation}
\frac{C_1}{N} \leq \mathbb{E}W_2^2(f_N, g_N) \leq \frac{C_2}{N}
\end{equation}
where $C_1$ and $C_2$ only depend on $c$.

The density functions $f_N$ and $g_N$ have total mass $1+\bar{r}^N$ and must be rescaled to mass one in order to obtain the desired result.
 Recalling that $g = g/(1+\bar{r}^N)$, we can rescale the squared Wasserstein metric~\cite[Proposition~7.16]{Villani} to obtain
\[
W_2^2(f_N/\langle f_N \rangle, g) = \left(\frac{1}{1+\bar{r}^N}\right)^2  W_2^2(f_N, g_N) 
\]
where
\[  \left(\frac{1}{1+c}\right)^2 \leq \left(\frac{1}{1+\bar{r}^N}\right)^2 \leq \left(\frac{1}{1-c}\right)^2.\]
Thus we conclude that $\mathbb{E}W_2^2(f_N/\langle f_N \rangle, g) = \bO(\frac{1}{N})$.
\end{proof}

\begin{rem}
The $L_2$ norm is significantly more sensitive to noise in this setting since $\mathbb{E}L_2^2(f_N,g_N)=\mathbb{E}||f_N-g_N||_2^2 = \mathbb{E}\left(\frac{1}{N}\sum_{i=1}^N|r_i|^2 \right)= O(1)$.
\end{rem}    

\subsection{Higher dimensions}

The analysis of the Wasserstein metric becomes much more difficult in higher dimensions. However, we can still analyze the effects of noise through the computation of an upper bound on the metric.

From the definition of the quadratic Wasserstein metric~\eqref{eq:W2}, it is clear that any transport map $T$ satisfies the inequality
		\[W_2^2(f,g)\leq \int\left|x-T(x)\right|^2f(x)dx. \]
     
     Consider the following two-dimensional example on the domain $\Omega = [0,1]\times[0,1]$ with the constant density function $g = 1$. Consider the noise function $r$ such that for each $(x,y)\in\Omega$, $r(x,y)$ is a random variable with uniform distribution on $[-c , c]$, $0<c<1$. We define the noisy density function $f = g + r$ and assume that \(\int_{\Omega}f = \int_{\Omega}g \).
     
    We use the Wasserstein metric to measure the difference between $g$ and its noisy version $f$. Since strong convergence in distribution implies convergence of the Wasserstein metric, we can approximate the density function $f$ by the piecewise constant function ${f_N}$ for the convenience of calculation. 
\[
f_N(x,y) = 1 + r_{ij},\quad x_i=\frac{i}{N} < x \leq\frac{i+1}{N}=x_{i+1},\quad y_j=\frac{j}{N} < y \leq\frac{j+1}{N}=y_{j+1}.
\]

\begin{figure}
\centering
\subfigure[]{\includegraphics[width=0.45\textwidth]{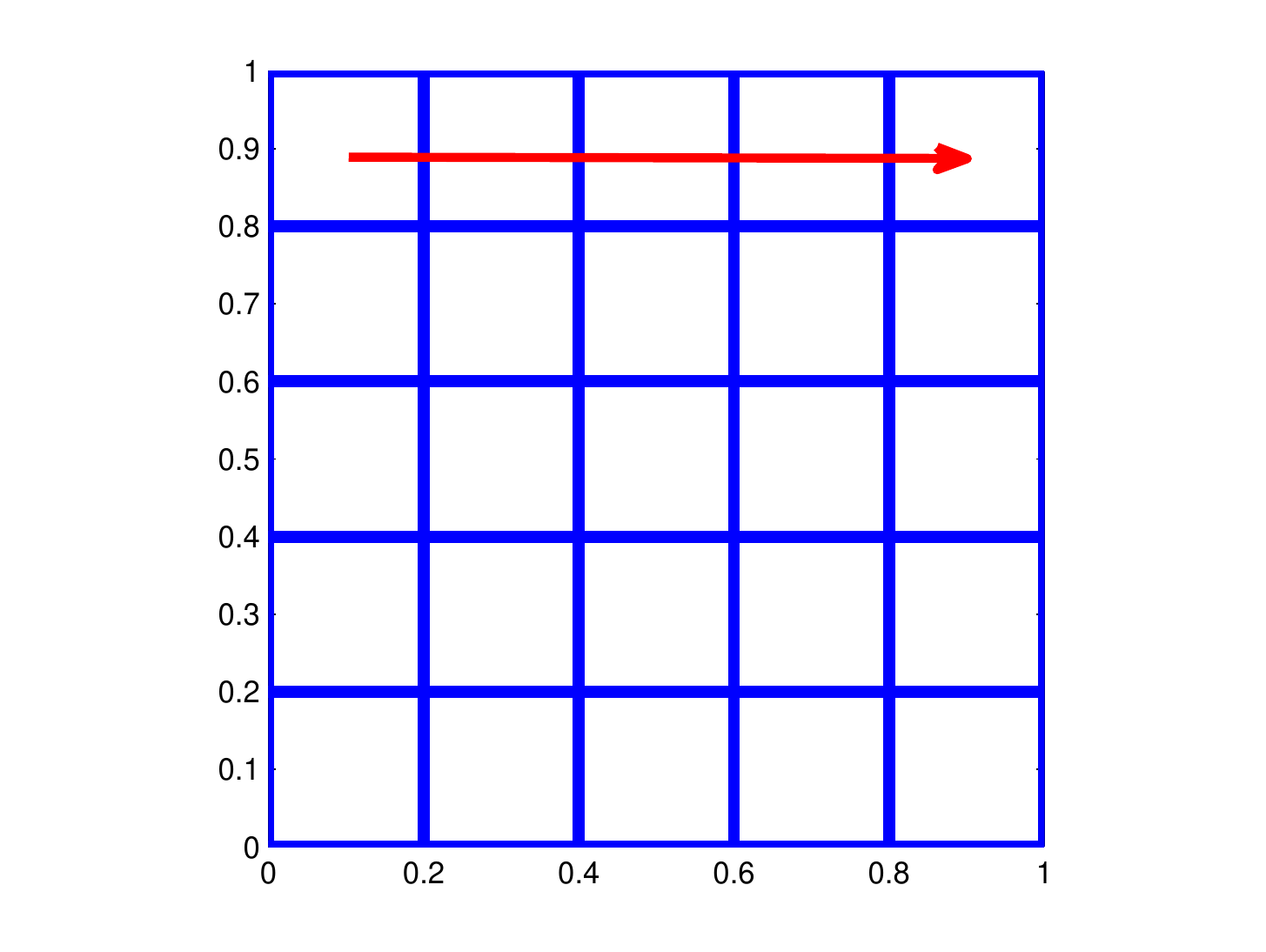}\label{fig:T1}}
\subfigure[]{\includegraphics[width=0.45\textwidth]{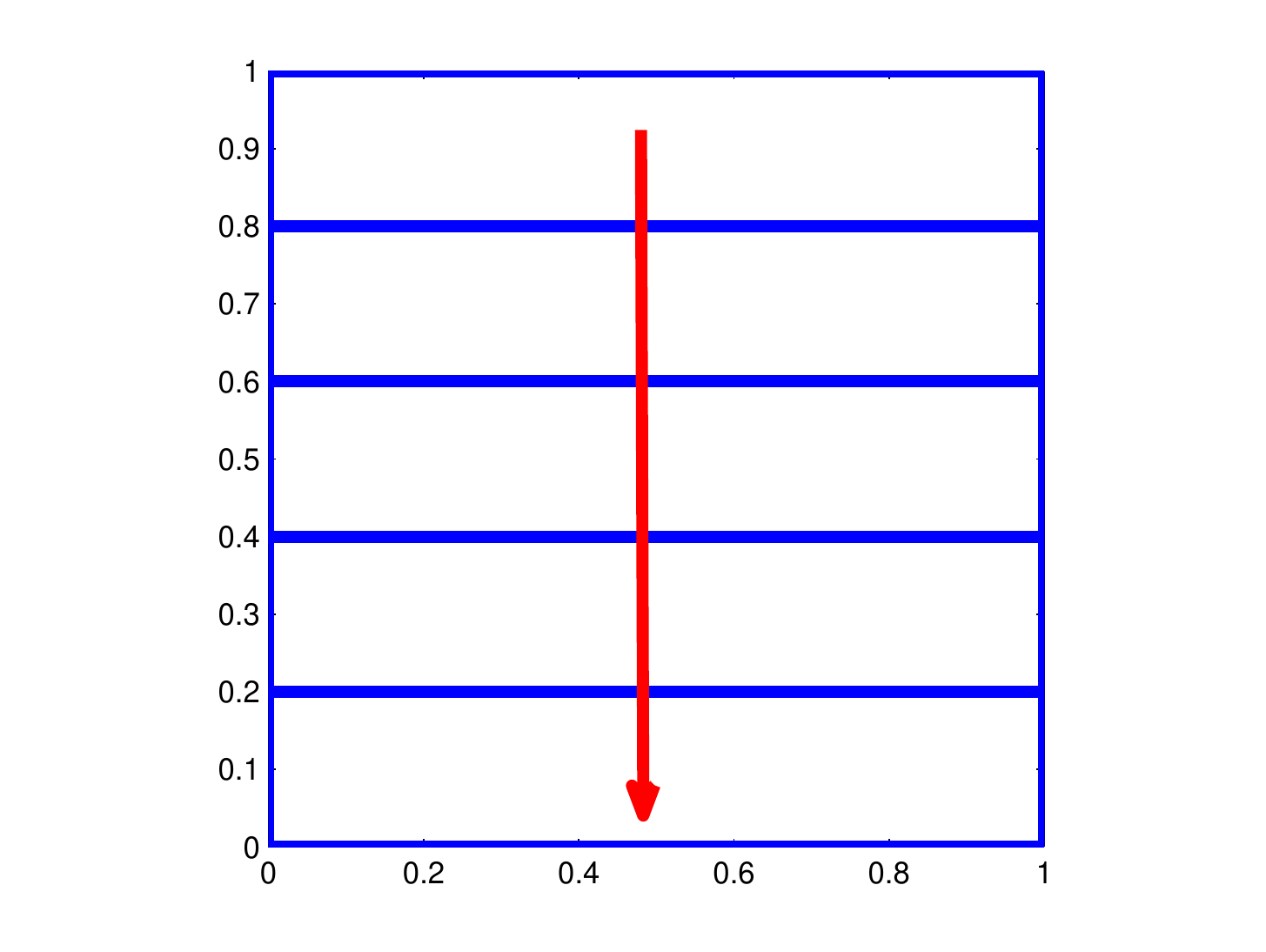}\label{fig:T2}}
\caption{\subref{fig:T1}~The optimal map for each row: $T_x = T_i$ for $x_i < x \leq x_{i+1}$ and  \subref{fig:invcdf}~the optimal map in y direction: $T_y$}
\label{fig:T}
\end{figure}

One approach to rearranging all the mass from $f_N$ to $g$ is to define $T$ in two steps as in Figure~\ref{fig:T}. First, with $y$ fixed, one can find the optimal map $T_x$ that averages each row. This is equivalent to a 1D optimal transport problem. Each row $i$ is mapped into a uniform density after rearrangement by the optimal map $T_i$. Secondly, with $x$ fixed, one can average the the density values of all the rows. Again, this is a 1D optimal transport problem and we have an explicit form for the optimal map $T_y$. The resulting map $T_N$ that rearranges $f_N$ to $g$ is $T_y \circ T_x$. Here $T_x = T_i$ for $x_{i-1} < x \leq x_i$, $i=1,\dots,N$.

The rearrangement determined by $T_N$ is not optimal, but does provide an upper bound on the value of the Wasserstein metric:
\bq
\mathbb{E} W_2^2(f_N,g) \leq  
\mathbb{E}\left(\iint\left(\left|x-T_x(x)\right|^2+\left|y-T_y(y)\right|^2\right)f_N(\bold{x})dxdy\right) \propto  \bO\left(\frac{1}{N}\right).
\eq

Finally, by the Lebesgue dominated convergence theorem, 
\bq
\mathbb{E} W_2^2(f,g) \leq \lim_{N\rightarrow \infty}   \mathbb{E}\left(\int_{\R^2}\left|\bold{x}-T_N(\bold{x})\right|^2f_N(\bold{x})d\bold{x}\right)=0.
\eq

For higher dimensions $n \geq 3$, we can similarly reduce the problem to several 1D optimal transport problems. The ultimate goal is to find a particular map that is not optimal, but that provides an upper bound that goes to zero as the mesh is refined.

\section{Numerical Computation of the Wasserstein metric}\label{sec:numerics}
We are interested in computing the  Wasserstein metric between two distributions $f$, $g$, which are supported on a rectangle $X$.  This can be accomplished via the solution of the \MA equation with non-homogeneous Neumann boundary conditions:
\bq\label{eq:MAA}
\begin{cases}
\det(D^2u(x)) = f(x)/g(\nabla u(x)) + \langle u \rangle,& x\in X\\
\nabla u(x) \cdot \nu = x\cdot \nu, & x \in \partial X\\
u \text{ is convex.}
\end{cases}
\eq
\begin{rem}
The Neumann boundary condition is easily generalised to the situation where $f$ and $g$ are supported on different rectangles~\cite{FroeseTransport}.
\end{rem}

The squared Wasserstein metric is then given by
\bq\label{eq:WassMA}
W_2^2(f,g) = \int_X f(x)\abs{x-\nabla u(x)}^2\,dx.
\eq

We solve the \MA equation numerically using an almost-monotone finite difference method relying on the following reformulation of the \MA operator, which automatically enforces the convexity constraint~\cite{FroeseTransport}.
\begin{multline}\label{eq:MA_convex}
{\det}^+(D^2u) = \\ \min\limits_{\{v_1,v_2\}\in V}\left\{\max\{u_{v_1,v_1},0\} \max\{u_{v_2,v_2},0\}+\min\{u_{v_1,v_1},0\} + \min\{u_{v_2,v_2},0\}\right\}
\end{multline}
where $V$ is the set of all orthonormal bases for $\R^2$.

Equation~\eqref{eq:MA_convex} can be discretized by computing the minimum over finitely many directions $\{\nu_1,\nu_2\}$, which may require the use of a wide stencil.  For simplicity and brevity, we describe a compact version of the scheme and refer to~\cite{FroeseTransport,FOFiltered} for complete details.

We begin by introducing the finite difference operators
\begin{align*}
[\Dt_{x_1x_1}u]_{ij} &= \frac{1}{dx^2} 
\left(
{u_{i+1,j}+u_{i-1,j}-2u_{i,j}}
\right)
\\
[\Dt_{x_2x_2}u]_{ij} &= \frac{1}{dx^2}
\left(
u_{i,j+1}+u_{i,j-1}-2u_{i,j}
\right)
\\
[\Dt_{x_1}u]_{ij} &= \frac{1}{2dx}
\left(
u_{i+1,j}-u_{i-1,j}
\right)\\
[\Dt_{x_2}u]_{ij} &= \frac{1}{2dx}
\left(
u_{i,j+1}-u_{i,j-1}
\right)\\
[\Dt_{vv}u]_{ij} &= \frac{1}{2dx^2}\left(u_{i+1,j+1}+u_{i-1,j-1}-2u_{i,j}\right)\\
[\Dt_{\vp\vp}u]_{ij} &= \frac{1}{2dx^2}\left(u_{i+1,j-1}+u_{i+1,j-1}-2u_{i,j}\right)\\
[\Dt_{v}u]_{ij} &= \frac{1}{2\sqrt{2}dx}\left(u_{i+1,j+1}-u_{i-1,j-1}\right)\\
[\Dt_{\vp}u]_{ij} &= \frac{1}{2\sqrt{2}dx}\left(u_{i+1,j-1}-u_{i-1,j+1}\right).
\end{align*}

In the compact version of the scheme, the minimum in~\eqref{eq:MA_convex} is approximated using only two possible values.  The first uses directions aligning with the grid axes.
\begin{multline}\label{MA1}
MA_1[u] = \max\left\{\Dt_{x_1x_1}u,\delta\right\}\max\left\{\Dt_{x_2x_2}u,\delta\right\} \\+ \min\left\{\Dt_{x_1x_1}u,\delta\right\} + \min\left\{\Dt_{x_2x_2}u,\delta\right\} - f / g\left(\Dt_{x_1}u, \Dt_{x_2}u\right) - u_0.
\end{multline}
Here $dx$ is the resolution of the grid, $\delta>K\Delta x/2$ is a small parameter that bounds second derivatives away from zero, $u_0$ is the solution value at a fixed point in the domain, and $K$ is the Lipschitz constant in the $y$-variable of $f(x)/g(y)$.

For the second value, we rotate the axes to align with the corner points in the stencil, which leads to
\begin{multline}\label{MA2}
MA_2[u] = \max\left\{\Dt_{vv}u,\delta\right\}\max\left\{\Dt_{\vp\vp}u,\delta\right\} + \min\left\{\Dt_{vv}u,\delta\right\} + \min\left\{\Dt_{\vp\vp}u,\delta\right\}\\ - f / g\left(\frac{1}{\sqrt{2}}(\Dt_{v}u+\Dt_{\vp}u), \frac{1}{\sqrt{2}}(\Dt_{v}u-\Dt_{\vp}u)\right) - u_0.
\end{multline}

Then the compact monotone approximation of the \MA equation is
\bq\label{eq:MA_compact} M_M[u] \equiv -\min\{MA_1[u],MA_2[u]\} = 0. \eq
We also define a second-order non-monotone approximation, obtained from a standard centred difference discretisation,
\bq\label{eq:MA_nonmon} M_N[u] \equiv -\left((\Dt_{x_1x_1}u)(\Dt_{x_2x_2}u)-(\Dt_{x_1x_2}u^2)\right) + f/g\left(\Dt_{x_1}u,\Dt_{x_2}u\right) + u_0 = 0.\eq

These are combined into an almost-monotone filtered approximation of the form
\bq\label{eq:MA_filtered} M_F[u] \equiv M_M[u] + \epsilon S\left(\frac{M_N[u]-M_M[u]}{\epsilon}\right) = 0 \eq
where $\epsilon$ is a small parameter and the filter $S$ is given by
\bq\label{eq:filter}
S(x) = \begin{cases}
x & \abs{x} \leq 1 \\
0 & \abs{x} \ge 2\\
-x+ 2  & 1\le x \le 2 \\
-x-2  & -2\le x\le -1.
\end{cases} 
\eq

The Neumann boundary condition is implemented using standard one-sided differences.

Once the discrete solution $u_h$ is computed, the squared Wasserstein metric is approximated via
\bq\label{eq:WassDiscrete}  W_2^2(f,g) \approx \sum\limits_{j=1}^m (x_j-D_{x_j}u_h)^T\diag(f)(x_j-D_{x_j}u_h). \eq

The computation of the discrete solution of~\eqref{eq:MA_filtered} requires the solution of a large system of nonlinear algebraic equations.  This is accomplished using Newton's method, which requires the Jacobian of the discrete scheme.  The Jacobian of the filtered scheme can be expressed as
\bq\label{eq:JacFilter}
\nabla M_F[u] = \left(1-S'\left(\frac{M_N[u]-M_M[u]}{\epsilon}\right)\right)\nabla M_M[u] + S'\left(\frac{M_N[u]-M_M[u]}{\epsilon}\right)\nabla M_N[u].
\eq
The (formal) Jacobians of the monotone and non-monotone components are given by
\begin{align*}
\nabla_u M_1[u] &= \left(\max\{\Dt_{x_2x_2},\delta\}\one_{\Dt_{x_1x_1}>\delta} + \one_{\Dt_{x_1x_1}\leq\delta}\right)\Dt_{x_1x_1}\\
&\phantom{=} + \left(\max\{\Dt_{x_1x_1},\delta\}\one_{\Dt_{x_2x_2}>\delta} + \one_{\Dt_{x_2x_2}\leq\delta}\right)\Dt_{x_2x_2}\\
&\phantom{=}-\frac{f}{g\left(\Dt_{x_1} u,\Dt_{x_2}u\right)^2} \nabla g\left(\Dt_{x_1} u,\Dt_{x_2}u\right)\cdot(\Dt_{x_1},\Dt_{x_2}) - \one_{x=x_0},\\
\nabla_u M_2[u] &= \left(\max\{\Dt_{\vp\vp},\delta\}\one_{\Dt_{vv}>\delta} + \one_{\Dt_{vv}\leq\delta}\right)\Dt_{vv}\\
&\phantom{=} + \left(\max\{\Dt_{vv},\delta\}\one_{\Dt_{\vp\vp}>\delta} + \one_{\Dt_{\vp\vp}\leq\delta}\right)\Dt_{\vp\vp}\\
&\phantom{=}-{f}/{g\left(\frac{1}{\sqrt{2}}(\Dt_{v}u+\Dt_{\vp}u), \frac{1}{\sqrt{2}}(\Dt_{v}u-\Dt_{\vp}u)\right)^2}\\
&\phantom{==} \nabla g\left(\frac{1}{\sqrt{2}}(\Dt_{v}u+\Dt_{\vp}u), \frac{1}{\sqrt{2}}(\Dt_{v}u-\Dt_{\vp}u)\right)\\
&\phantom{==}\cdot\left(\frac{1}{\sqrt{2}}(\Dt_{v}+\Dt_{\vp}), \frac{1}{\sqrt{2}}(\Dt_{v}-\Dt_{\vp})\right) - \one_{x=x_0},\\
\nabla_u M_M[u] &= -\one_{M_M[u]=-M_1[u]}\nabla_u M_1[u] - \one_{M_M[u]=-M_2[u]} \nabla_u M_2u], \\
\nabla_{u} M_N[u] &= -(\Dt_{x_2x_2}u)\Dt_{x_1x_1} - (\Dt_{x_1x_1}u)\Dt_{x_2x_2} + 2(\Dt_{x_1x_2}u)\Dt_{x_1x_2} \\ 
&\phantom{=}+\frac{f}{g\left(\Dt_{x_1} u,\Dt_{x_2}u\right)^2} \nabla g\left(\Dt_{x_1} u,\Dt_{x_2}u\right)\cdot(\Dt_{x_1},\Dt_{x_2}) + \one_{x=x_0}.
\end{align*}

The availability of these Jacobians will become key in~\autoref{sec:linear}, where we will use these results to compute the Frechet gradient of the Wasserstein metric.  This, in turn, is needed for the minimization in the computational examples of~\autoref{sec:results}.

\section{Computation of Frechet Gradient}\label{sec:linear}

Our goal is to minimise the Wasserstein metric between computed data $f(v)$ and observed data $g$, where $f$ depends on a set of parameters $v$.  In order to do this efficiently, we will require the gradient of the squared Wasserstein metric with respect to the unknown parameters.

Our main focus here is computation of the Fr\'{e}chet gradient of the squared Wasserstein metric with respect to the data $f$, which is new in the context of full waveform inversion.  The gradient needed for the minimization is then obtained through the composition
\[ \nabla_f W_2^2(f(v)) \nabla_v f(v). \]
 As long as $\nabla_f W_2^2$ can be computed efficiently, techniques such as the adjoint state method can be used to efficiently construct the required gradient~\cite{Plessix}. 

In the present work, our focus is on the use of optimal transportation techniques, rather than on the use of a particular forward model for producing the data $f(v)$.  In the computations of \autoref{sec:results}, we will present the minimization for problems involving several different models.  For simplicity, and to keep the focus on the properties of the Wasserstein metric, we will simply use a forward difference approximation to estimate $\nabla_vf$.

Two different approaches are possible here.  One option is to directly linearize the Wasserstein metric, then discretize the result.  A second approach, which we pursue here, is to linearize the discrete approximation of the Wasserstein metric.  A key advantage of this approach is that it allows us to make use of the Jacobian~\eqref{eq:JacFilter} that is already being constructed in the process of solving the \MA equation.  We also argue that this is the correct gradient since our approach to full waveform inversion is exactly solving the optimization problem~\eqref{eq:fwi} where the misfit function~$d(f,g)$ is given by a discrete approximation to the squared Wasserstein metric.

Using the finite difference matrices introduced in \autoref{sec:numerics}, we can express the discrete Wasserstein metric as
\bq\label{eq:wassDisc}
d(f,g) = \sum\limits_{j=1}^n (x_j-D_{x_j}u_f)^T\diag(f)(x_j-D_{x_j}u_f)
\eq
where the potential $u_f$ satisfies the discrete \MA equation
\[\text{M}[u_f] = 0.\]  

\begin{lemma}[Frechet gradient of discrete Wasserstein metric]
The Frechet gradient of the discretized Wasserstein metric~\eqref{eq:wassDisc} is given by
\[ \nabla_f d(f,g) = \sum\limits_{j=1}^n\left[-2\nabla M_F^{-1}[u_f]^TD_{x_j}^T\diag(f) + \diag(x_j-D_{x_j}u_f)\right](x_j - D_{x_j}u_f). \]
\end{lemma}

\begin{proof}
The first variation of the squared Wasserstein metric as
\[ \delta d = -2\sum\limits_{j=1}^n (D_{x_j} \delta u)^T \diag(f) (x_j - D_{x_j}u_f) + \sum\limits_{j=1}^n (x_j-D_{x_j}u_f)^T \diag(\delta f) (x_j - D_{x_j}u_f). \]

Linearizing the \MA equation, we have to first order
\[ \nabla M_F[u_f] \delta u = \delta f. \]
Here $\nabla M_F$ is the (formal) Jacobian of the discrete \MA equation, which is already being inverted in the process of solving the \MA equation via Newton's method~\eqref{eq:JacFilter}.  Then the gradient of the discrete squared Wasserstein metric can be expressed as
\[ \nabla_f d = \sum\limits_{j=1}^n\left[-2\nabla M_F^{-1}[u_f]^TD_{x_j}^T\diag(f) + \diag(x_j-D_{x_j}u_f)\right](x_j - D_{x_j}u_f). \]
\end{proof}

Notice that once the \MA equation itself has been solved, this gradient is easy to compute as it only requires the inversion of a single matrix that is already being inverted as a part of the solution of the \MA equation.

\section{Computational Results}

\label{sec:results}

In this section, we provide examples of the minimization of the Wasserstein metric between given data $g$ and a modeled signal $f(v)$ that depends on the unknown parameters $v$. Minimization is performed using the Matlab function fmincon, equipped with the gradient described in~\autoref{sec:linear}.

The wave equation~\eqref{eq:forward} is solved by using finite difference scheme for a defined initial wave field.

\begin{align*}
u_{n,m}^{l+1} = &-u_{n,m}^{l-1}+2u_{n,m}^{l}\\
 &+v_{n,m}^2\Delta t^2\left(\frac{u_{n+1,m}^{l}-2u_{n,m}^{l}+u_{n-1,m}^{l}}{\Delta x^2} + \frac{u_{n,m+1}^{l}-2u_{n,m}^{l}+u_{n,m-1}^{l}}{\Delta z^2}\right)
\end{align*}
with the initial conditions 
\[
u_{n,m}^{-1} = \\f(n\Delta x, m\Delta z),\quad u_{n,m}^{0} = f(n\Delta x, m\Delta z).
\]

Here $u_{n,m}^l$ is the wave field at the time $l\Delta t$ and at the spatial position $(n\Delta x, m\Delta z)$. $v_{n,m}$ is the velocity at $(n\Delta x, m\Delta z)$. The step size $\Delta t$ is chosen to satisfy the numerical stability condition:
\[
\min(\Delta x, \Delta z)>\sqrt{2}\Delta t \max(v).
\]

To ensure the data to be positive which is a requirement for objects in optimal transportation, we work with something akin to a  local amplitude by defining
\[ \tilde{f}(x,t) = \sqrt{\int_{t-\epsilon}^{t+\epsilon} u(x,0,s)^2\,ds}  \]
where $\epsilon=10\Delta t$.  Finally, this profile is normalised to produce a density function $f(x,t)$ that has unit mass. 

\subsection{Single layer model}
\label{sec:2para}
We first consider a material composed of a single layer of depth $h$ and velocity $v$.  
We define the data $f_{h,v}(s,t)$ to be the resulting data, which we obtain by solving the wave equation for $u_{h,v}$ and processing the results.

We consider the particular case of $h^* = 2$, $v^*=1$.
In order to define the target profile $g$, which mimics the observed data, we add noise $N(s,t)$ chosen uniformly at random from $[-M,M]$, 
\[ \tilde{g}(s,t) = \max\{u_{2,1}(s,t) + N(s,t),0\}, \]
where $M$ is approximately 2\% the maximum value of $f_{2,1}$.
See Figure~\ref{fig:layer1noisy10}.  Then our goal is to determine $h$ and $v$ that minimize
\[ W_2^2(f_{h,v}, g).  \]

We initialize with the guess $h = 2.5$ and $v = 1.75$ and perform minimization over the parameters $h$ and $v^{-1}$.  The convergence history is displayed in Figure~\ref{fig:convergence2}.  Despite the noise in the target profile, we recover the parameters $\tilde{h} = 2.2157$ and $\tilde{v}=1.0953$ after fifteen iterations, with a squared Wasserstein metric of $3.36\times10^{-4}$.  (The required stepsize in the minimization algorithm became too small to improve appreciably beyond this).   

For reference, we also compare the noisy target~$g$ with the exact signal~$f_{2,1}$ (without noise).  This yields a 
 a squared Wasserstein metric of $7.49\times10^{-4}$, so that the error in the recovered parameters can be explained by the noise.

\begin{figure}
\centering
\includegraphics[width=0.5\textwidth]{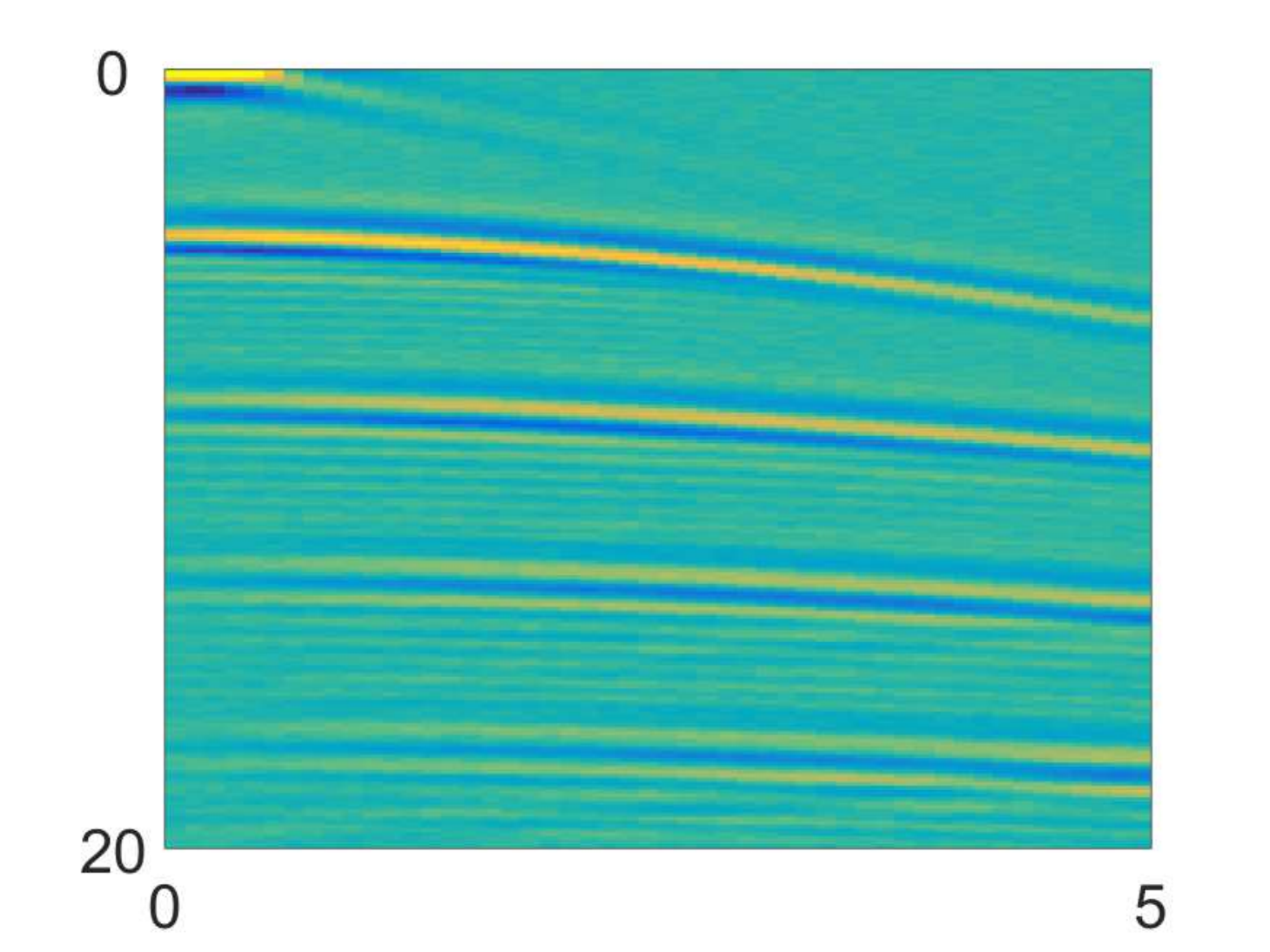}
\label{fig:layer1noisy10}
\caption{A signal produced from a single layer model with added noise.}
\end{figure}

\begin{figure}[htp]
\centering
\subfigure[]{\includegraphics[width=0.475\textwidth]{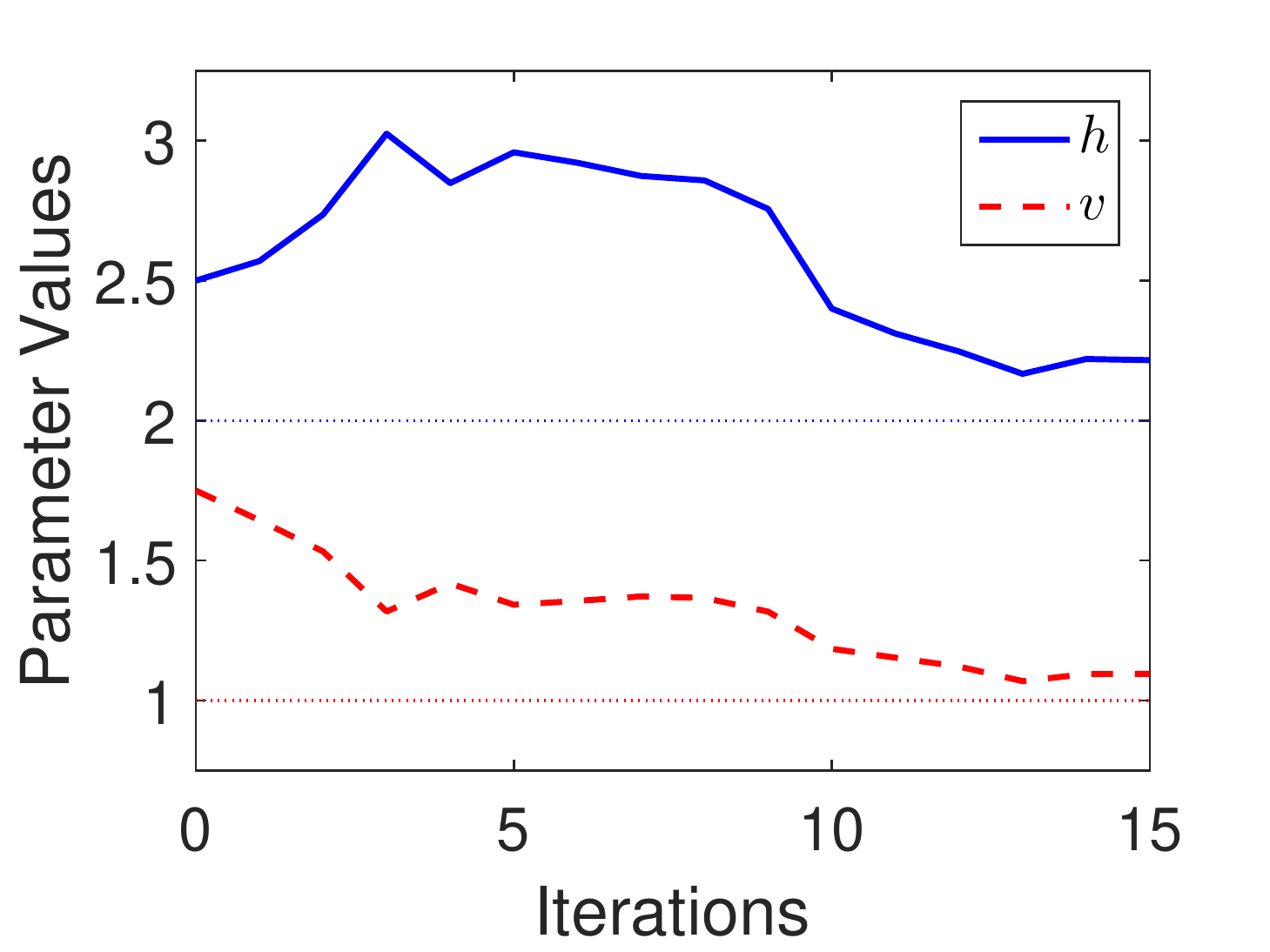}\label{fig:convergence1LayerPar}}
\subfigure[]{\includegraphics[width=0.475\textwidth]{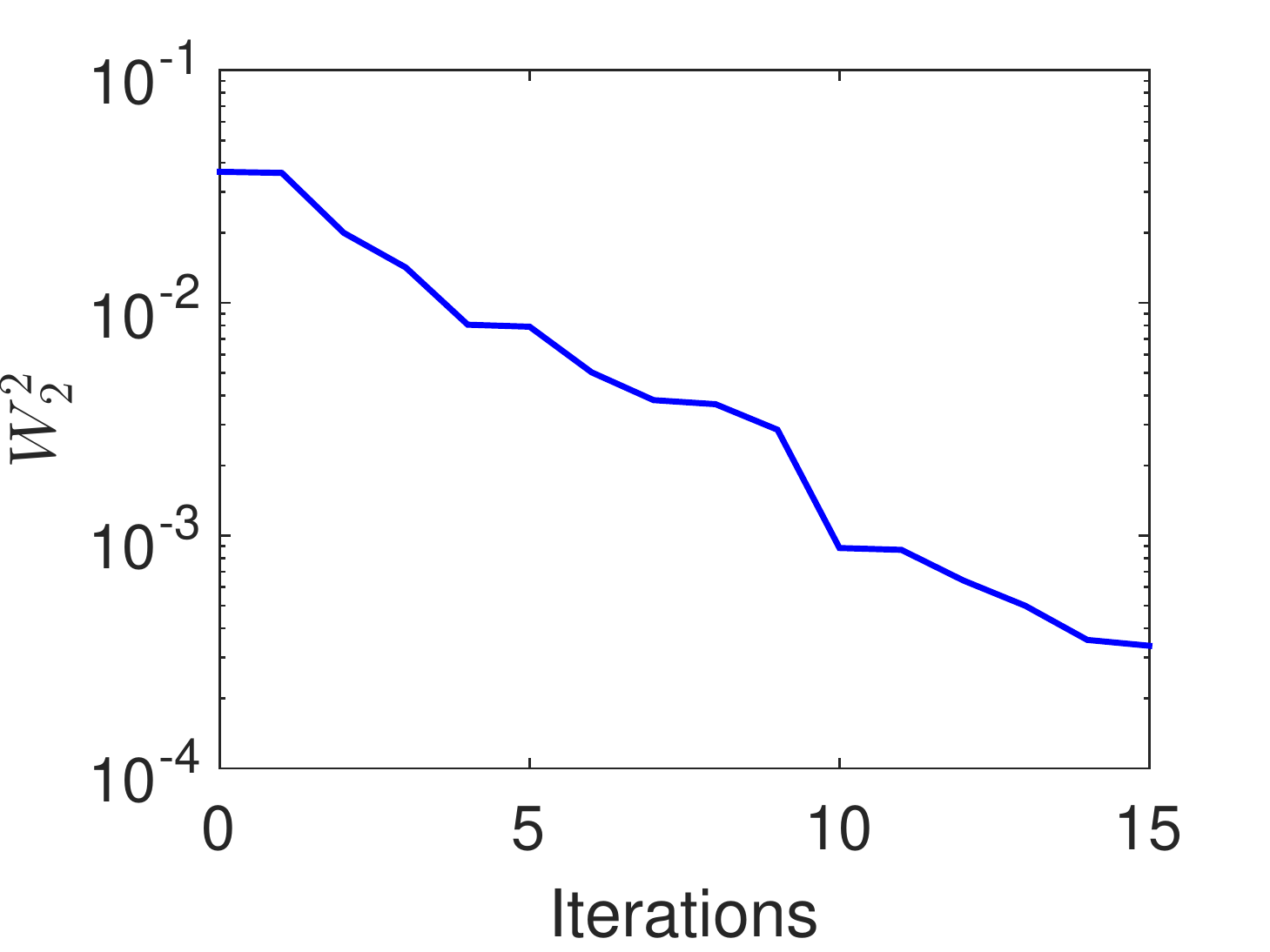}\label{fig:convergence1LayerW2}}
\caption{Convergence history for a single layer model.}
\label{fig:convergence1}
\end{figure}

\subsection{Two layer model}

Next, we consider the case where the material is composed of two different layers.  The top layer has depth $h_1$ and velocity $v_1$ while the bottom layer has depth $h_2$ and velocity $v_2$; see Figure~\ref{fig:commonshot}.  

We look at the particular case where the given target density $g$ is defined by the parameter values 
\[  h_1^* = 0.75, \quad v_1^* = 1, \quad h_2^* = 1, \quad v_2^* = 1.5. \]
As in the previous example, we add noise to this target.
The resulting signal is shown in Figure~\ref{fig:signal2layer}.

In this case, the distance $W_2$ depends on the four parameters $h_1,v_1^{-1},h_2,v_2^{-1}$.  We initialize with the guess $h_1 = 0.5$, $v_1 = 1.5$, $h_2 = 0.75$, and $v_2 = 2$.  
After 33 iterations, we recover the parameter values $\tilde{h}_1 = 0.772$, $\tilde{h}_2 = 0.991$, $\tilde{v}_1 = 1.0318$, and $\tilde{v}_2 = 1.519$ with a squared misfit value of $2.06\times10^{-5}$.  The convergence history is presented in Figure~\ref{fig:convergence2}.

As noted in~\cite{EFWass}, when the model involves both depth and velocity, the resulting distance can contain narrow valleys, and computing the minimum can require small stepsizes.  We were still able to effectively compute the minimum in this setting, but we expect that quasi-Newton methods would enable even faster convergence.

\begin{figure}[htp]
\centering
\subfigure[]{\includegraphics[width=0.45\textwidth]{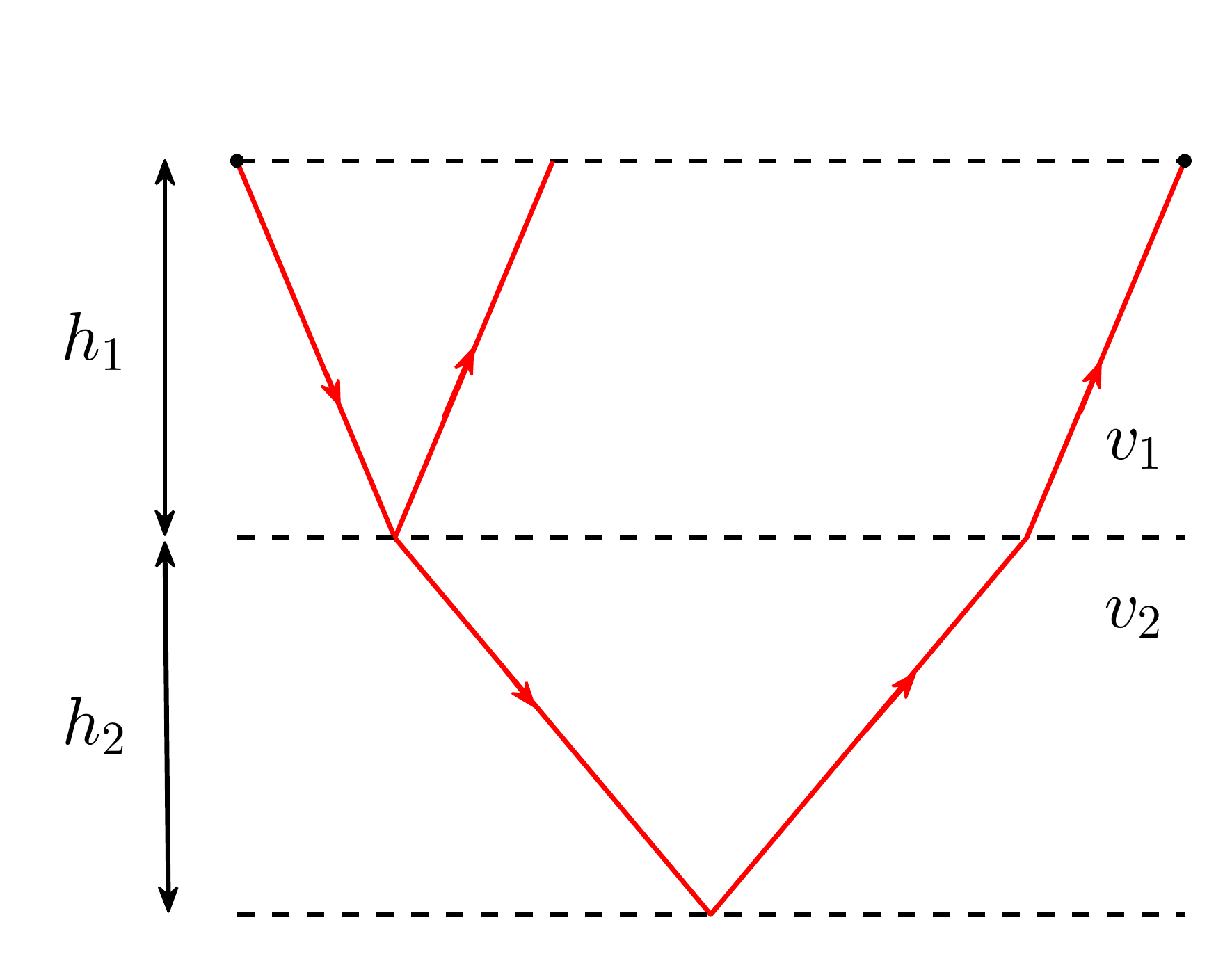}\label{fig:commonshot}}
\subfigure[]{\includegraphics[width=0.45\textwidth]{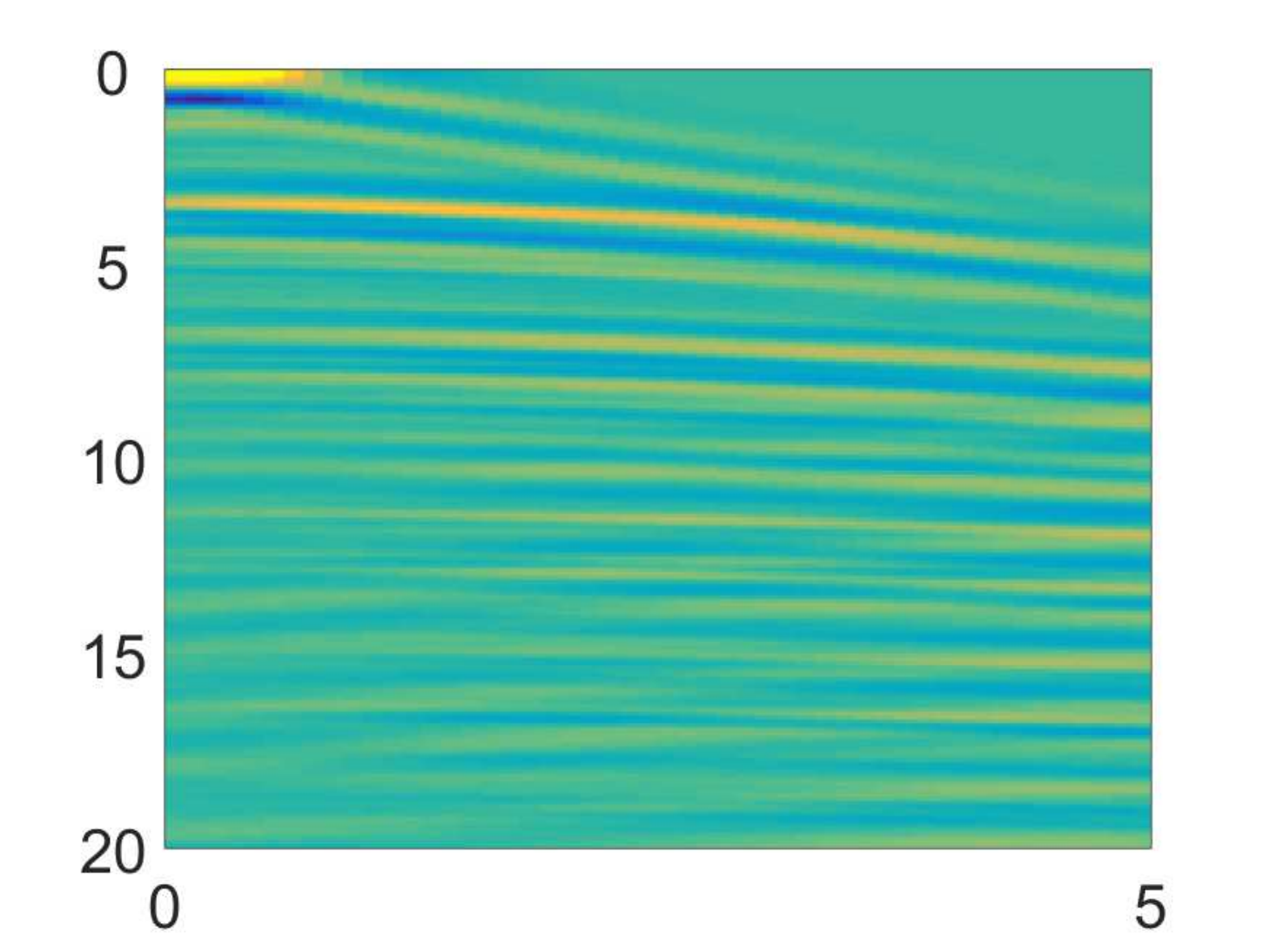}\label{fig:signal2layer}}
\caption{\subref{fig:commonshot}~A two-layer material and  \subref{fig:signal2layer}~the resulting signal.}
\label{fig:layer2}
\end{figure}

\begin{figure}[htp]
\centering
\subfigure[]{\includegraphics[height=0.45\textwidth]{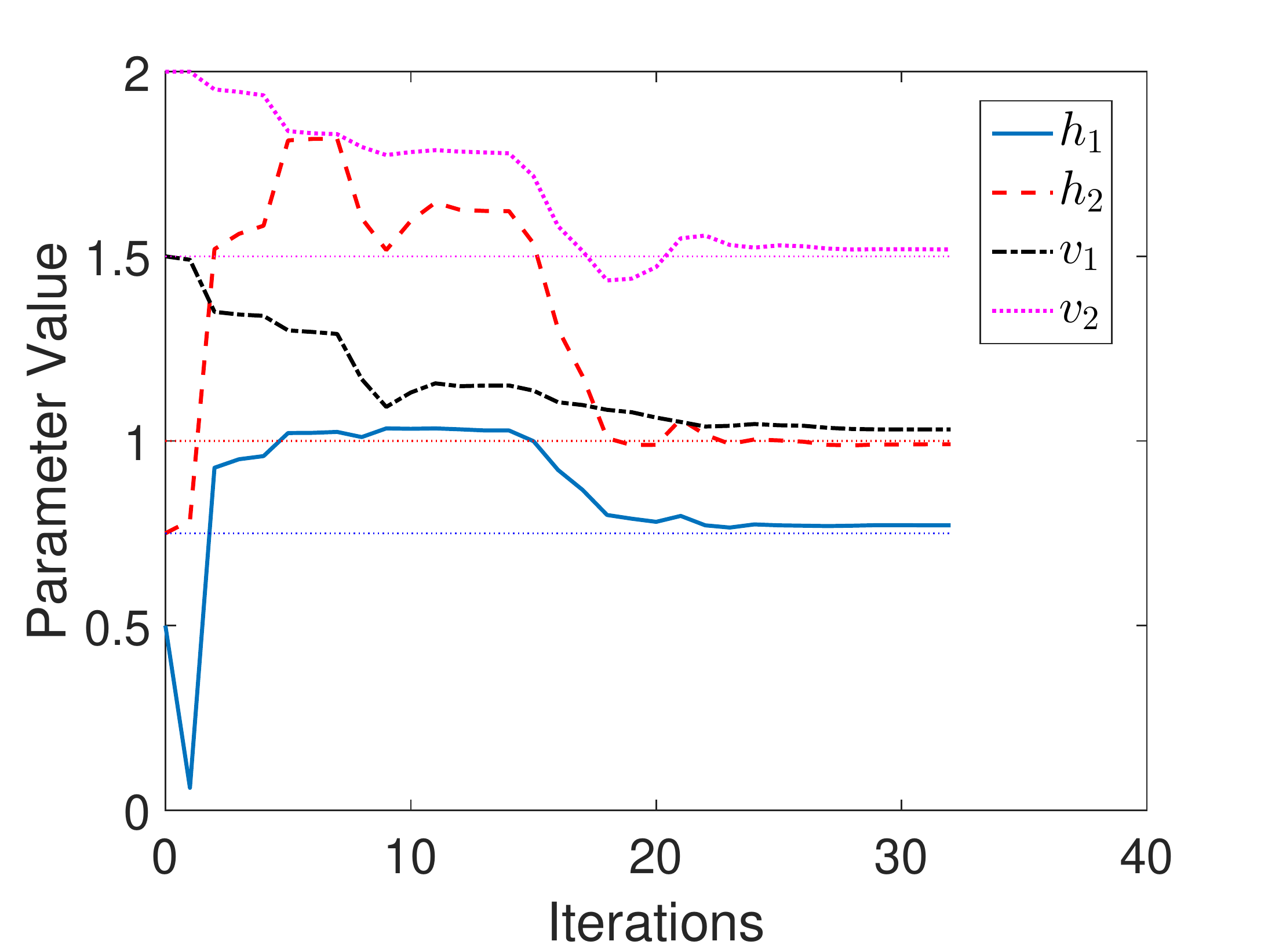}\label{fig:convergence2LayerPar}}
\subfigure[]{\includegraphics[height=0.45\textwidth]{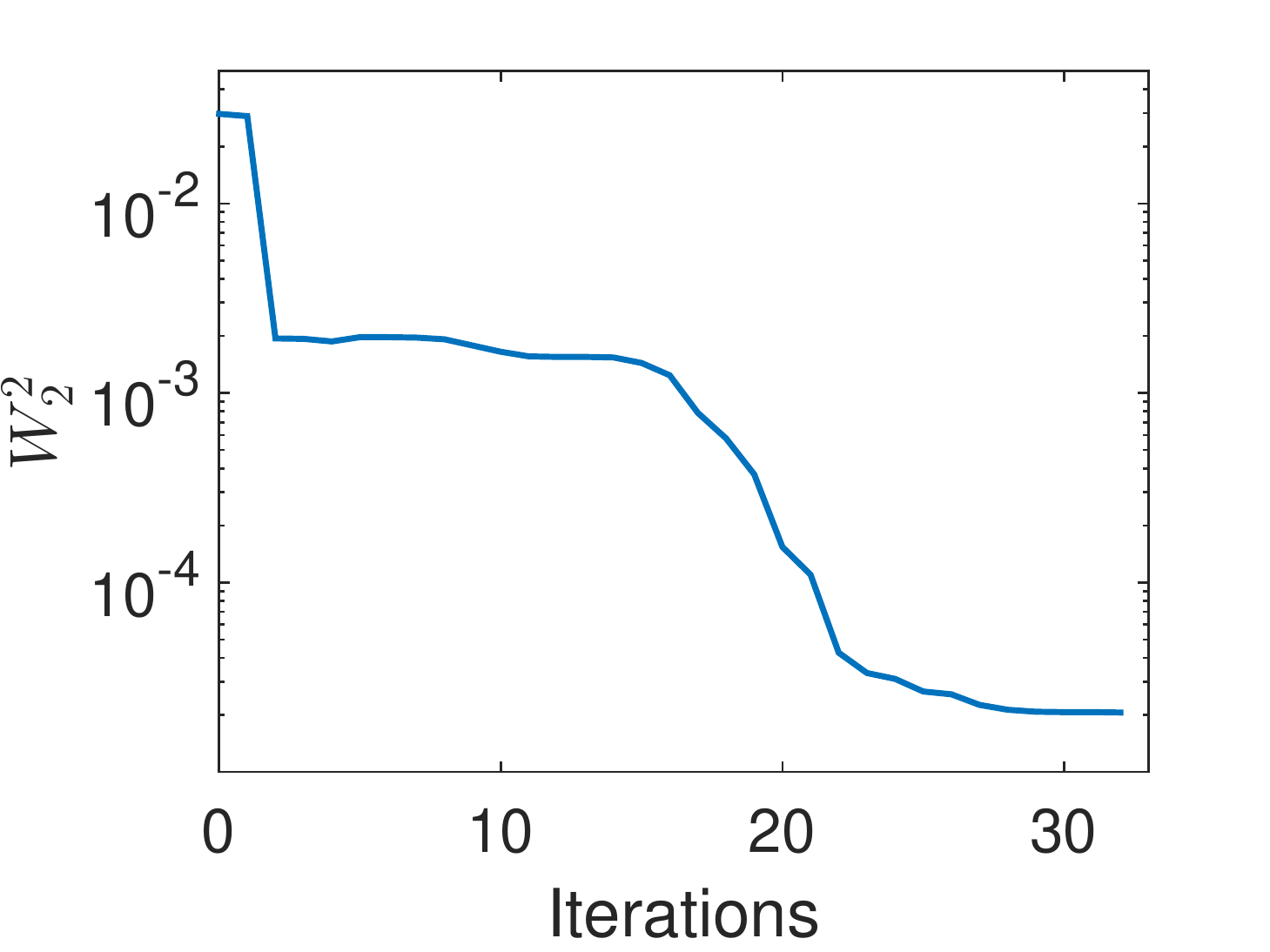}\label{fig:convergence2LayerW2}}
\caption{Convergence history for a two-layer material.}
\label{fig:convergence2}
\end{figure}

\subsection{Six Parameter model}
\label{sec:6para}
We next consider the case of a  piecewise constant material.  See Figure~\ref{fig:Param6} for the set-up.

We look at the particular case where the given target density $g$ is defined by the parameter values 
\[  v_1^* = 1, \quad v_2^* = 1.5, \quad v_3^* = 1, \quad
v_4^* = 2, \quad v_5^* = 2.5, \quad v_6^* = 1.75. \]
As in the previous example, we add noise to this target.
The resulting signal is shown in Figure~\ref{fig:signal6Param}.

\begin{figure}
\centering
\subfigure[]{\includegraphics[width=0.45\textwidth]{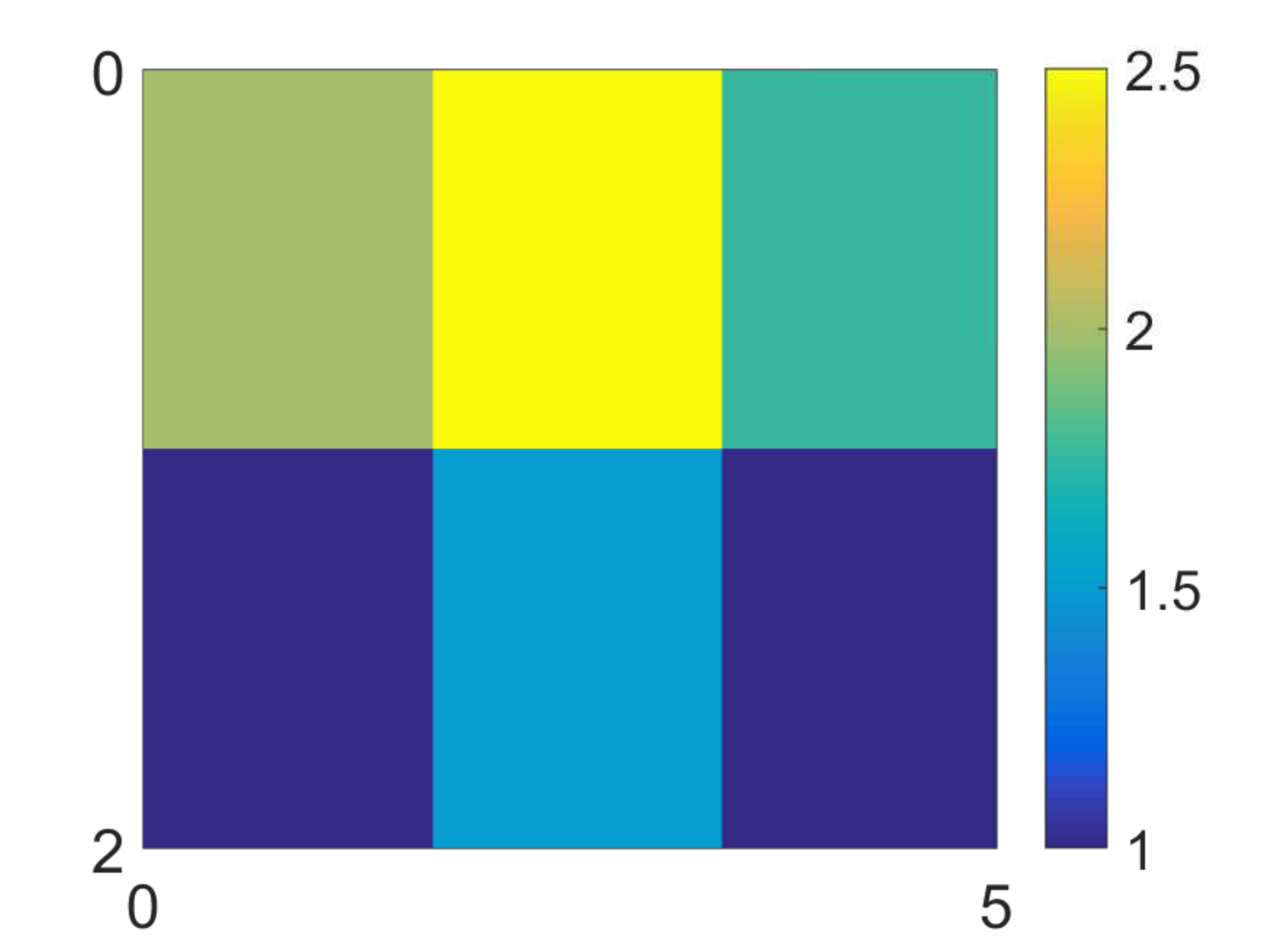}\label{fig:param6Vg}}
\subfigure[]{\includegraphics[width=0.45\textwidth]{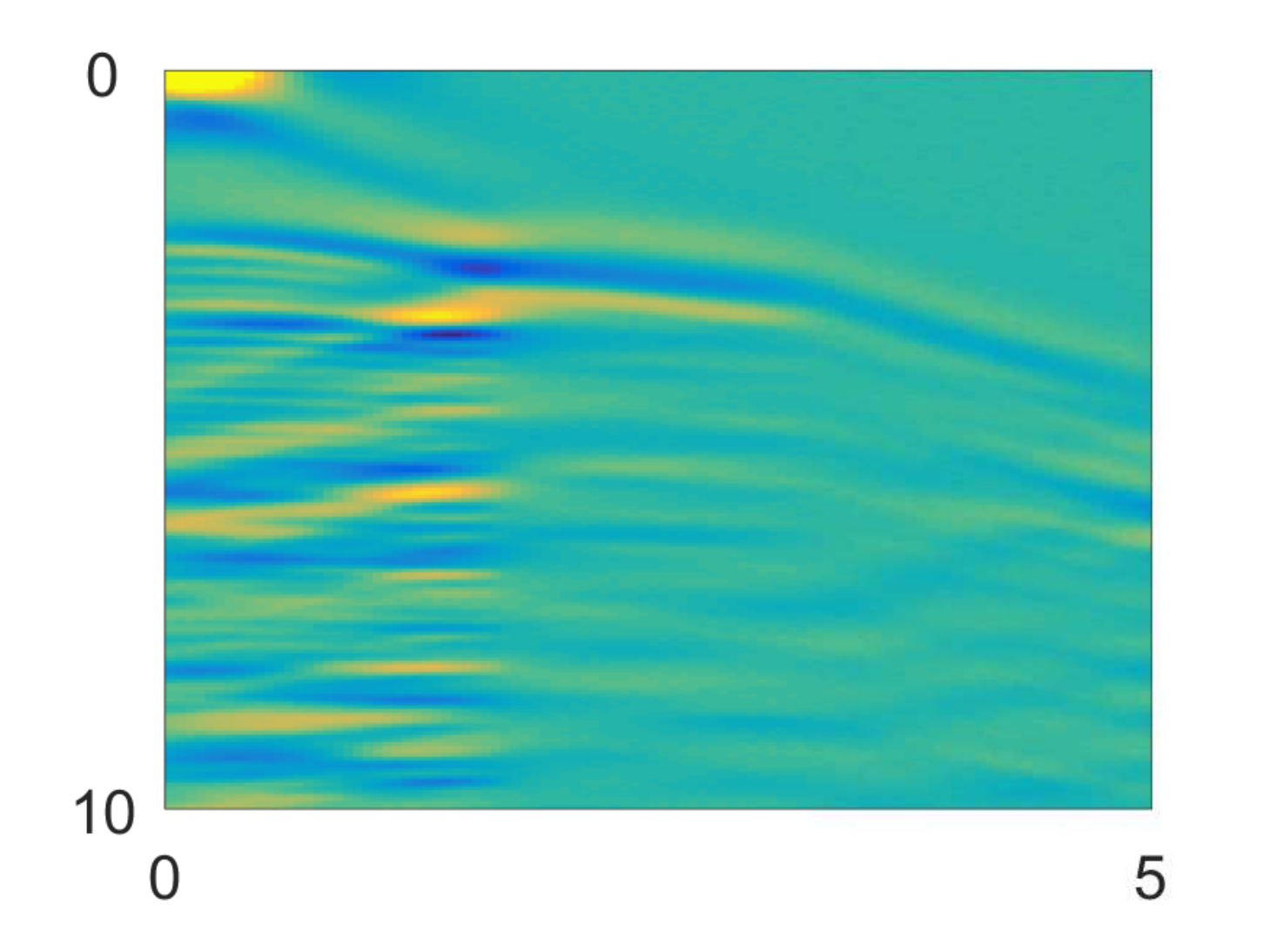}\label{fig:signal6Param}}
\subfigure[]{\includegraphics[width=0.45\textwidth]{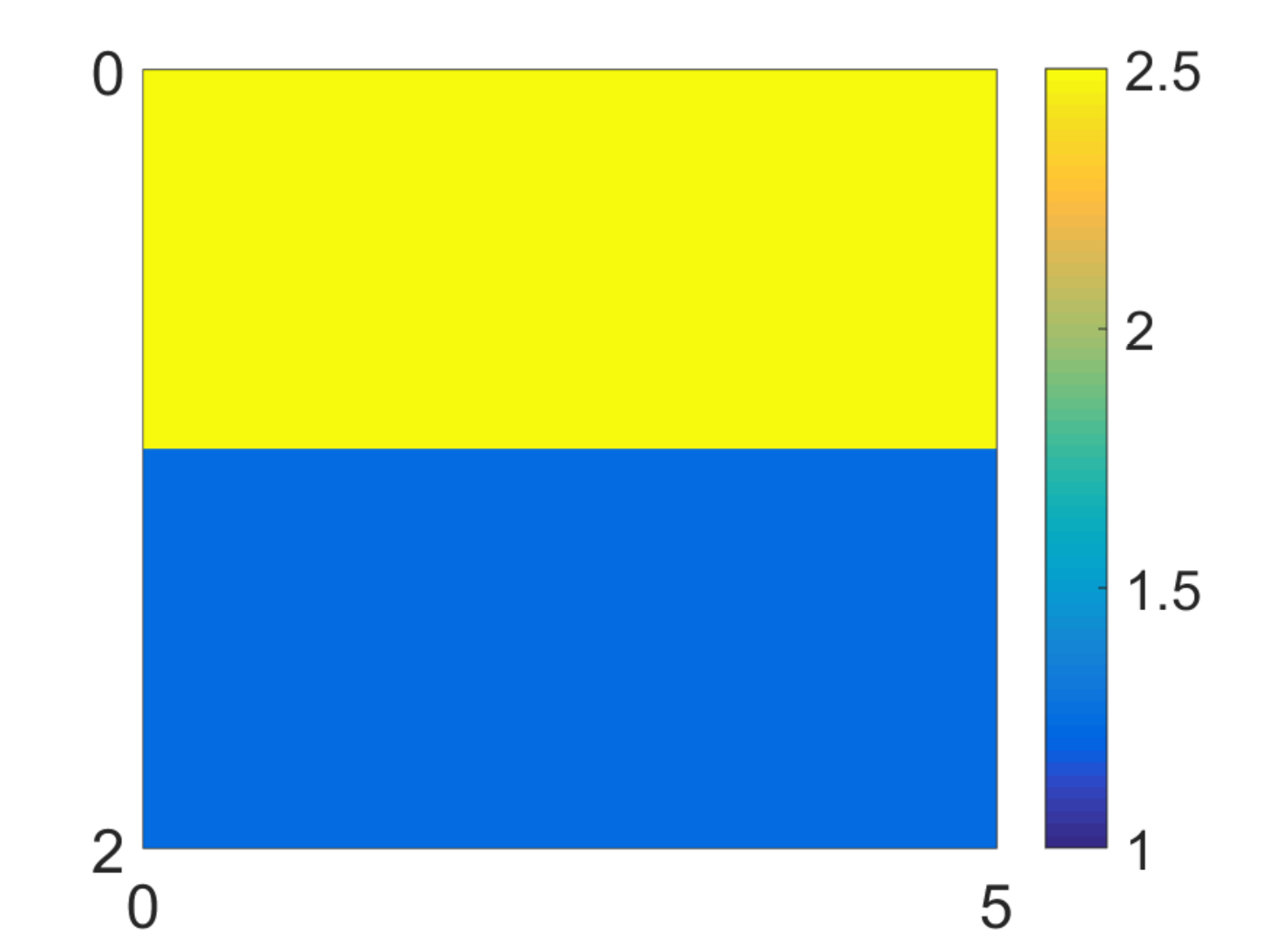}\label{fig:param6V0}}
\subfigure[]{\includegraphics[width=0.45\textwidth]{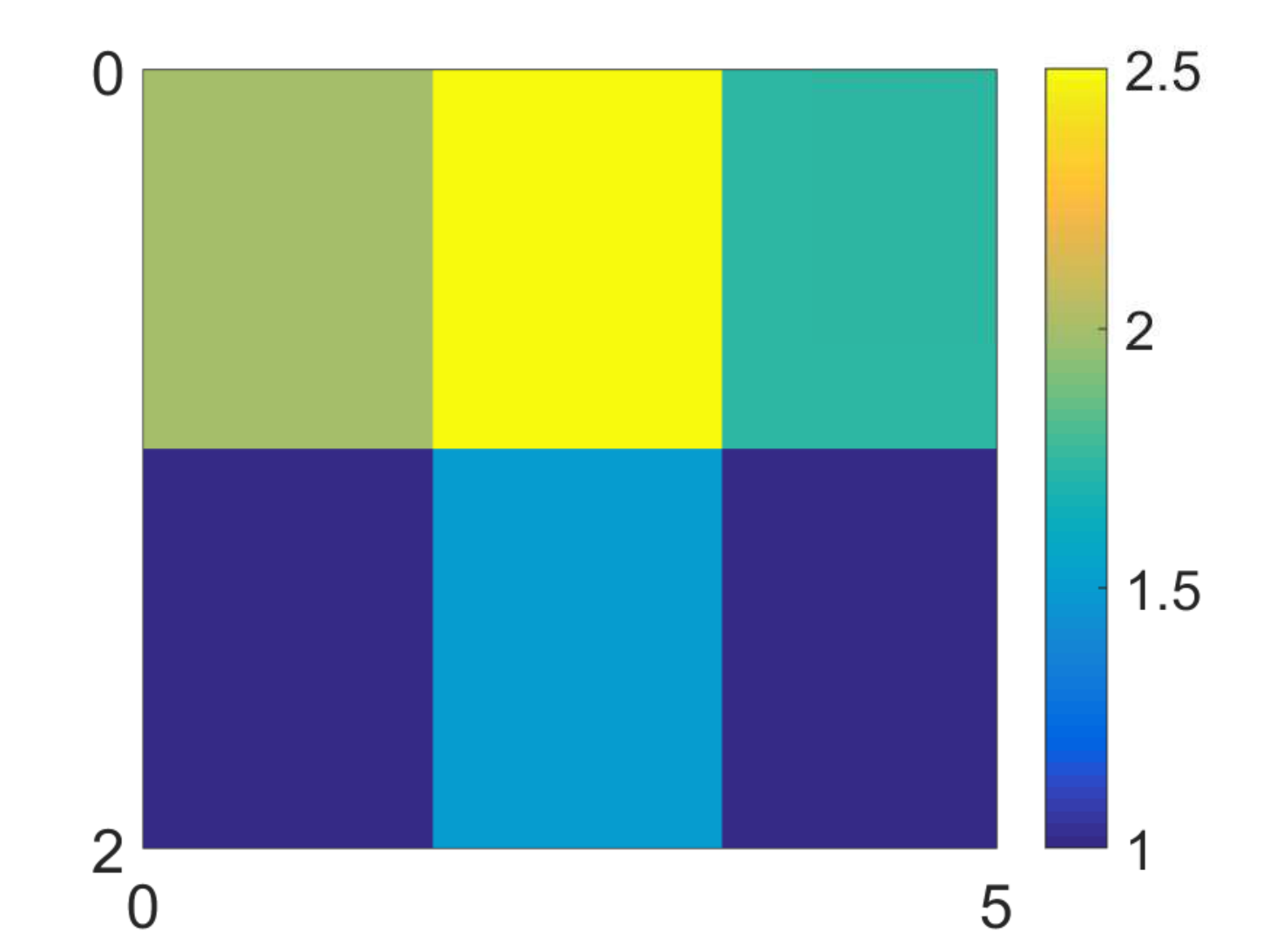}\label{fig:param6Vc}}
\caption{\subref{fig:param6Vg}~A six parameter velocity model used to generate  \subref{fig:signal6Param}~a target signal~$g$. \subref{fig:param6V0}~Initial and \subref{fig:param6Vc}~computed velocity.}
\label{fig:Param6}
\end{figure}

In this case, the distance $W_2$ depends on the six parameters $1/v_i$, $i=1,\ldots,6$.  We initialize with the guess $v_1=v_2=v_3=1.25$ and $v_4=v_5=v_6 = 2.5$.  
In this example, which depends only on velocity and not on depth, the convergence proceeds without the need for very small stepsizes that we observed in the previous example.  After 72 iterations, we recover the parameter values $\tilde{v}_1=1.0034$, $\tilde{v}_2=    1.5058$, $\tilde{v}_3=    0.9996$, $\tilde{v}_4=    1.9932$, $\tilde{v}_5=    2.4889$, and $\tilde{v}_6=    1.7296$ with a squared misfit value of $3.94\times10^{-6}$.   The convergence history is presented in Figure~\ref{fig:convergence6}.
For reference, comparison of the noisy target with the exact signal (without noise) yielded a squared Wasserstein metric of $4.82\times10^{-6}$.

\begin{figure}[htp]
\centering
\subfigure[]{\includegraphics[height=0.45\textwidth]{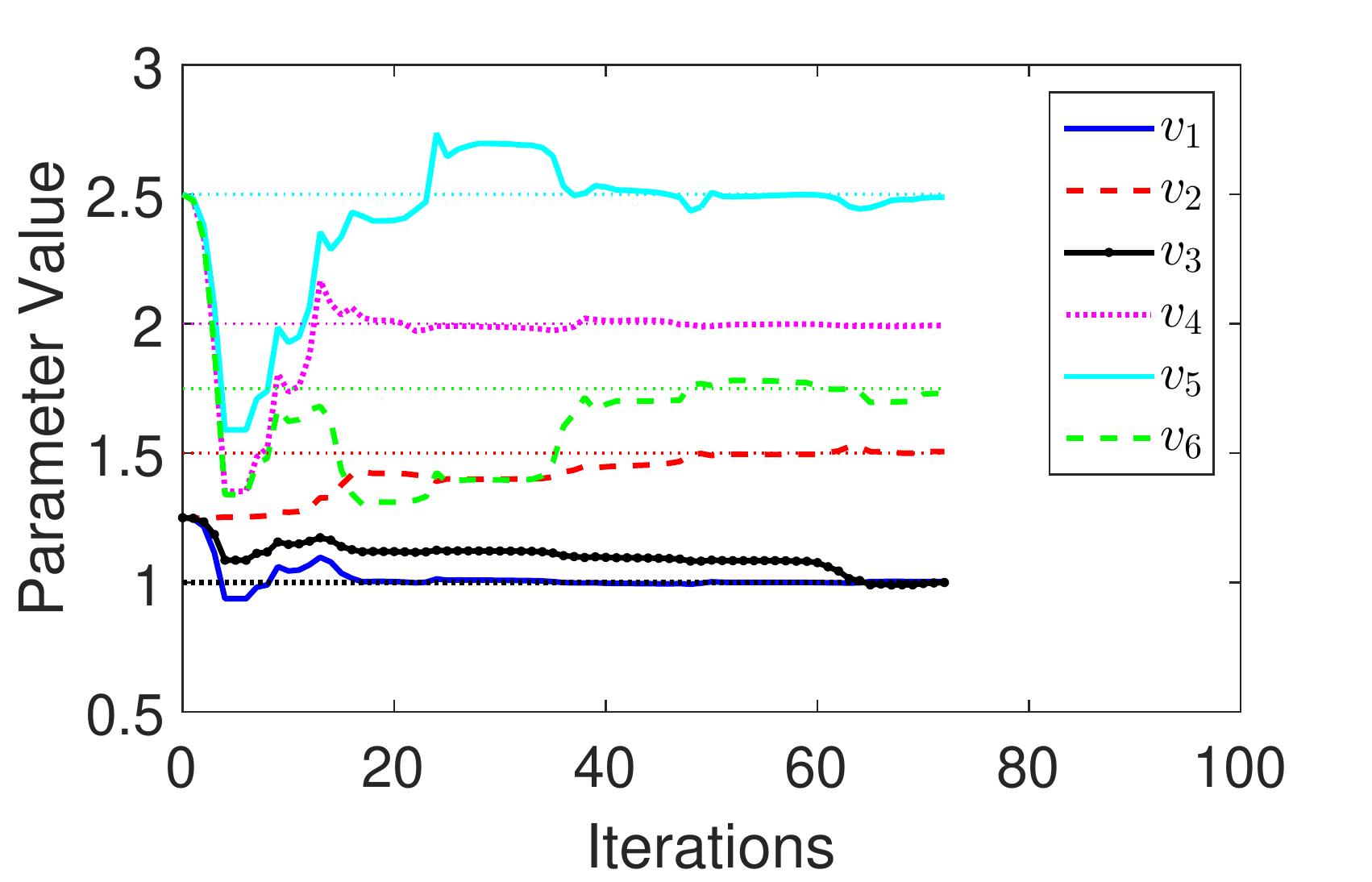}\label{fig:convergence6Param_Par}}
\subfigure[]{\includegraphics[height=0.45\textwidth]{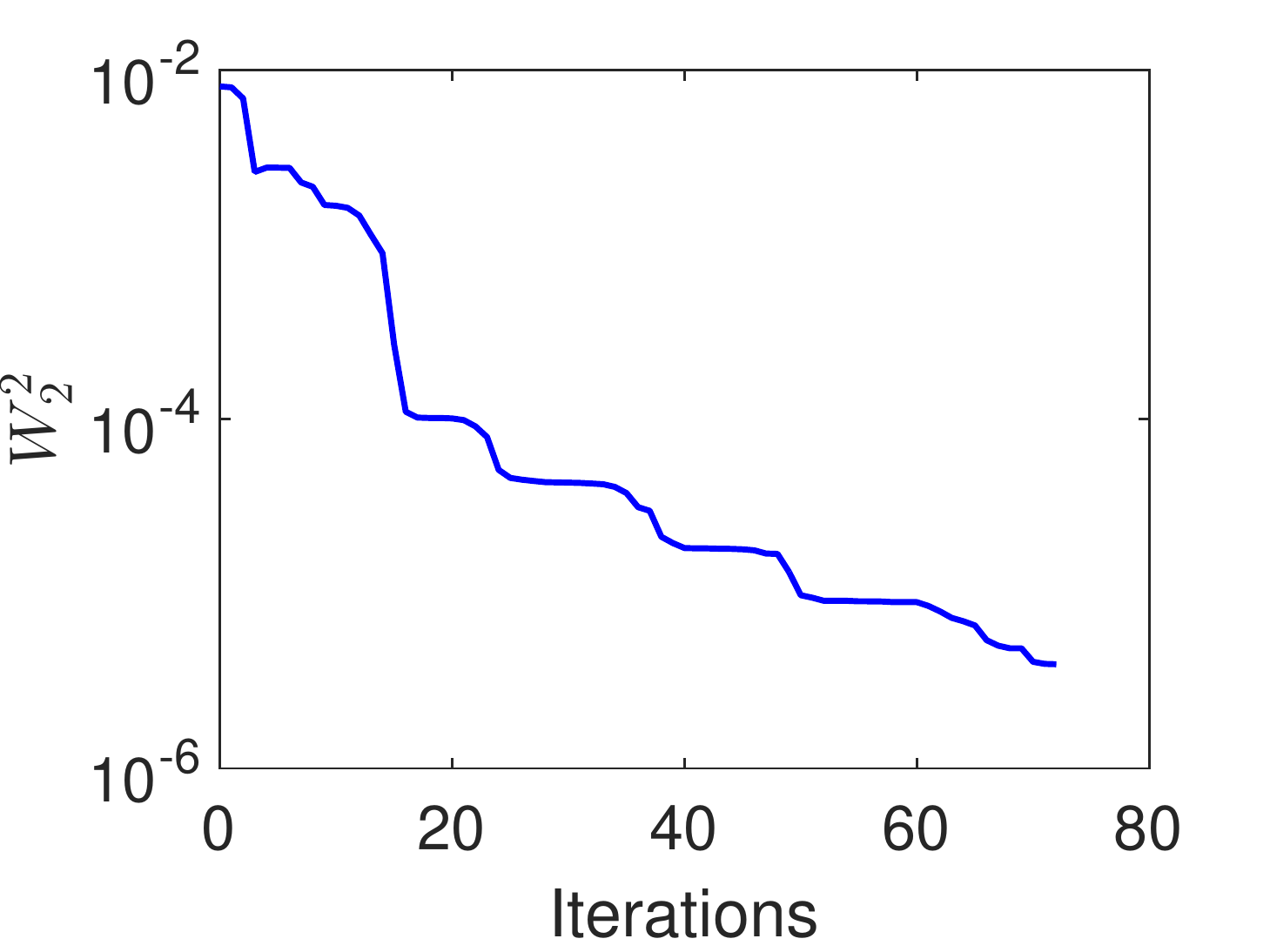}\label{fig:convergence6Param_W2}}
\caption{Convergence history for a six parameter model.}
\label{fig:convergence6}
\end{figure}

\subsection{Twelve Parameter model}

We again consider a piecewise constant velocity model, but this time increase the number of parameters to twelve.  See Figure~\ref{fig:Param12} for the set-up used to construct the (noisy) target density $g$, as well as the resulting signal.

\begin{figure}
\centering
\subfigure[]{\includegraphics[width=0.45\textwidth]{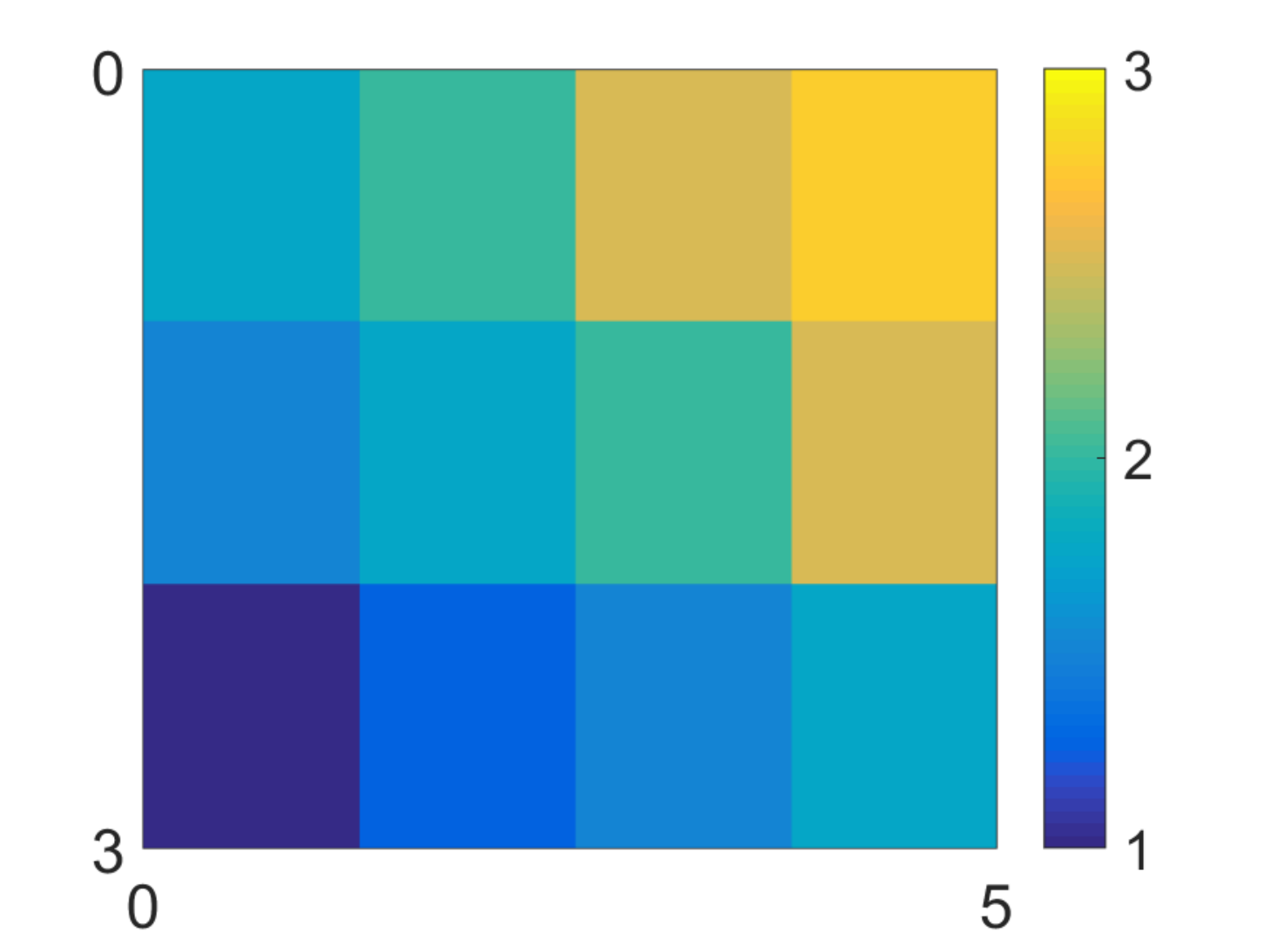}\label{fig:param12Vg}}
\subfigure[]{\includegraphics[width=0.45\textwidth]{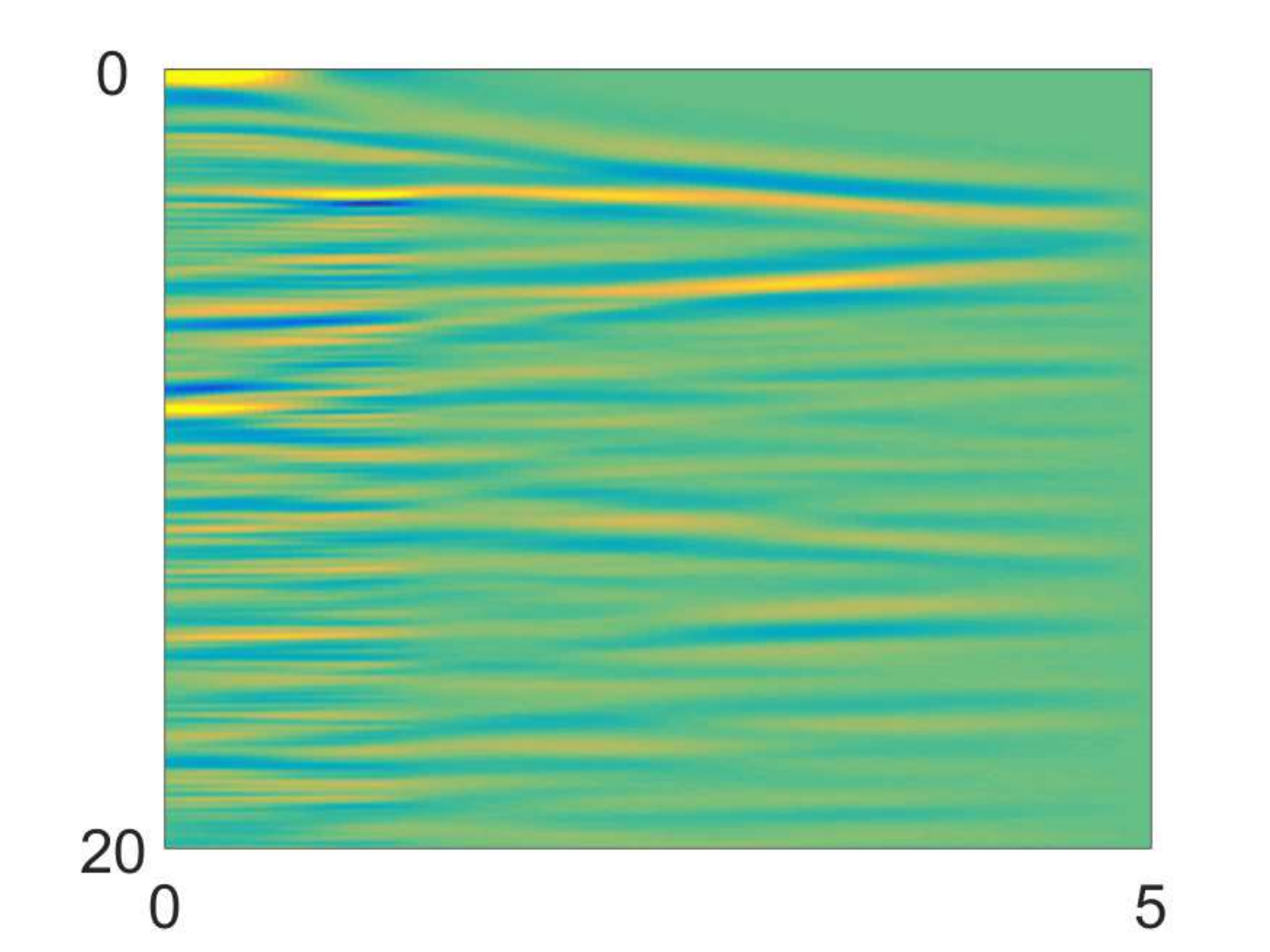}\label{fig:signal12Param}}
\subfigure[]{\includegraphics[width=0.45\textwidth]{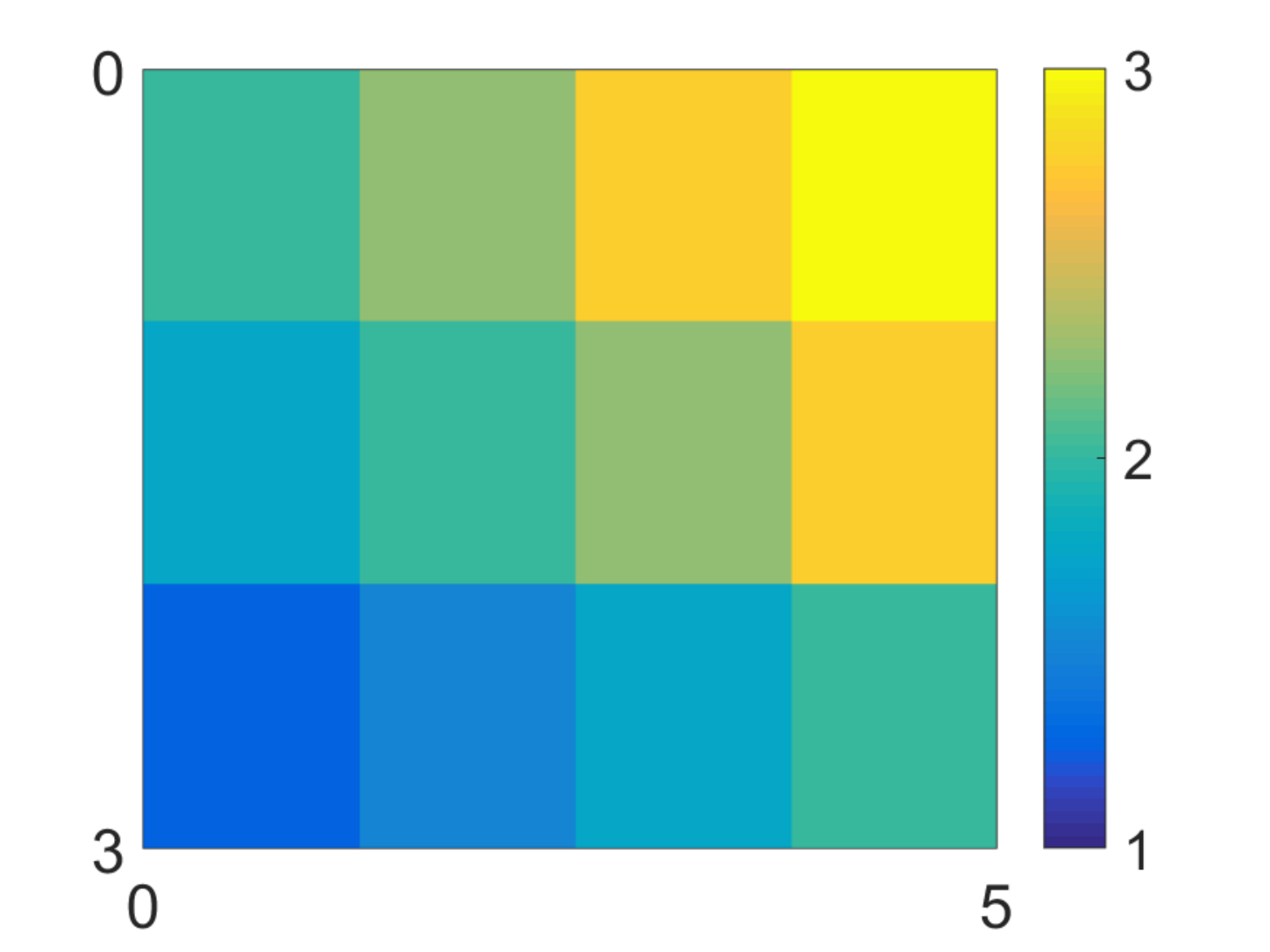}\label{fig:param12V0}}
\subfigure[]{\includegraphics[width=0.45\textwidth]{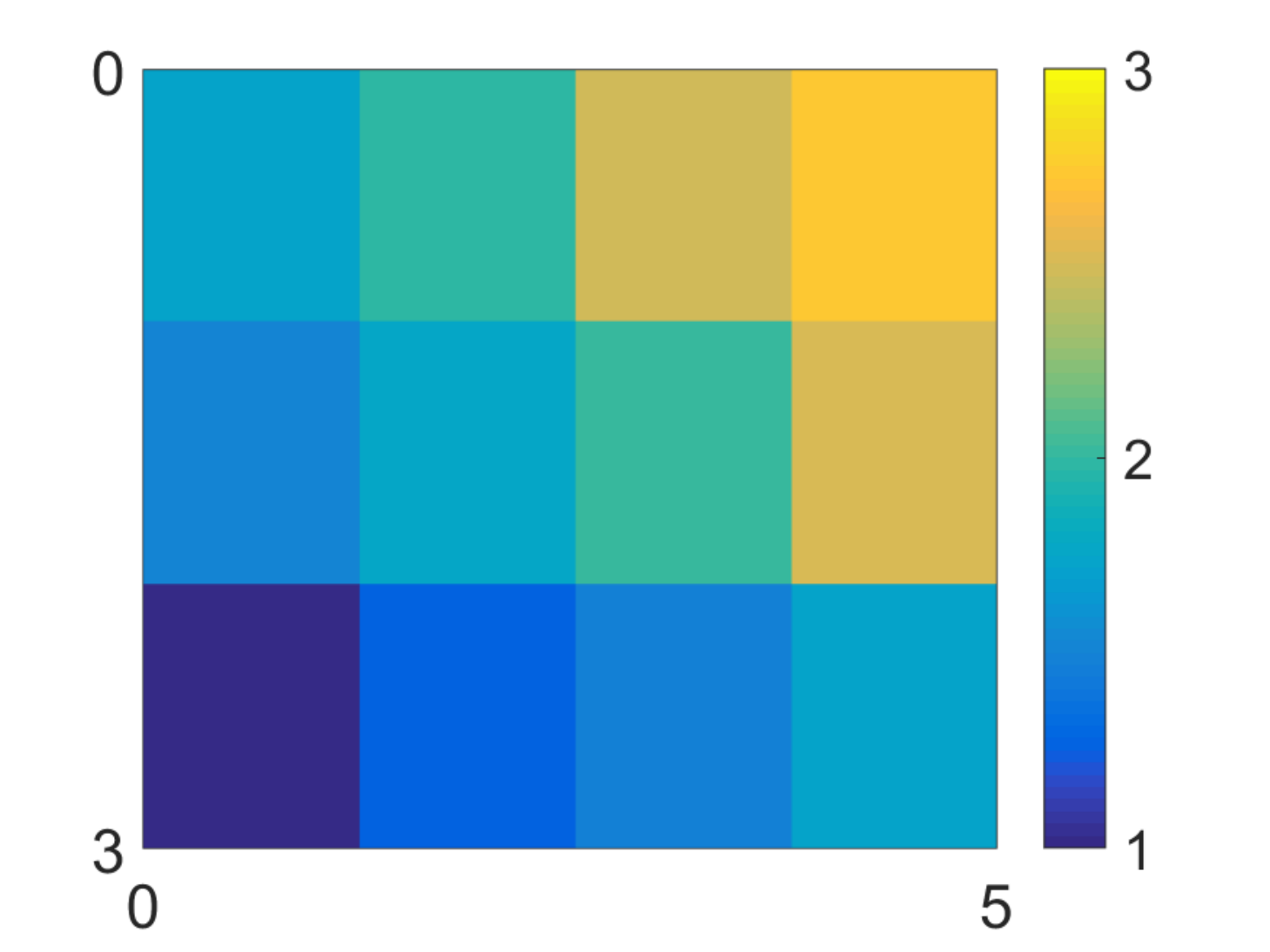}\label{fig:param12Vc}}
\caption{\subref{fig:param12Vg}~A twelve parameter velocity model used to generate  \subref{fig:signal12Param}~a target signal~$g$. \subref{fig:param12V0}~Initial and \subref{fig:param12Vc}~computed velocity.}
\label{fig:Param12}
\end{figure}

In this case, the distance $W_2$ depends on the twelve parameters $1/v_i$, $i=1,\ldots,12$.  We initialize with the guess $v = v^* + 0.25$.    After 132 iterations, we recover the twelve parameters with a maximum error of $\|\tilde{v}-v^*\|_\infty = 0.0091$ and a squared misfit value of $2.10\times10^{-6}$.   
For reference, comparison of the noisy target with the exact signal (without noise) yielded a squared Wasserstein metric of $3.16\times10^{-6}$, which suggests that the error in the recovered parameter values is due to noise in the data.
The convergence history is presented in Figure~\ref{fig:convergence12}. The simple models in~\autoref{sec:2para} and~\autoref{sec:6para} were included to indicate how the result depend on model complexity.

\begin{figure}[htp]
\centering
\subfigure[]{\includegraphics[width=0.475\textwidth]{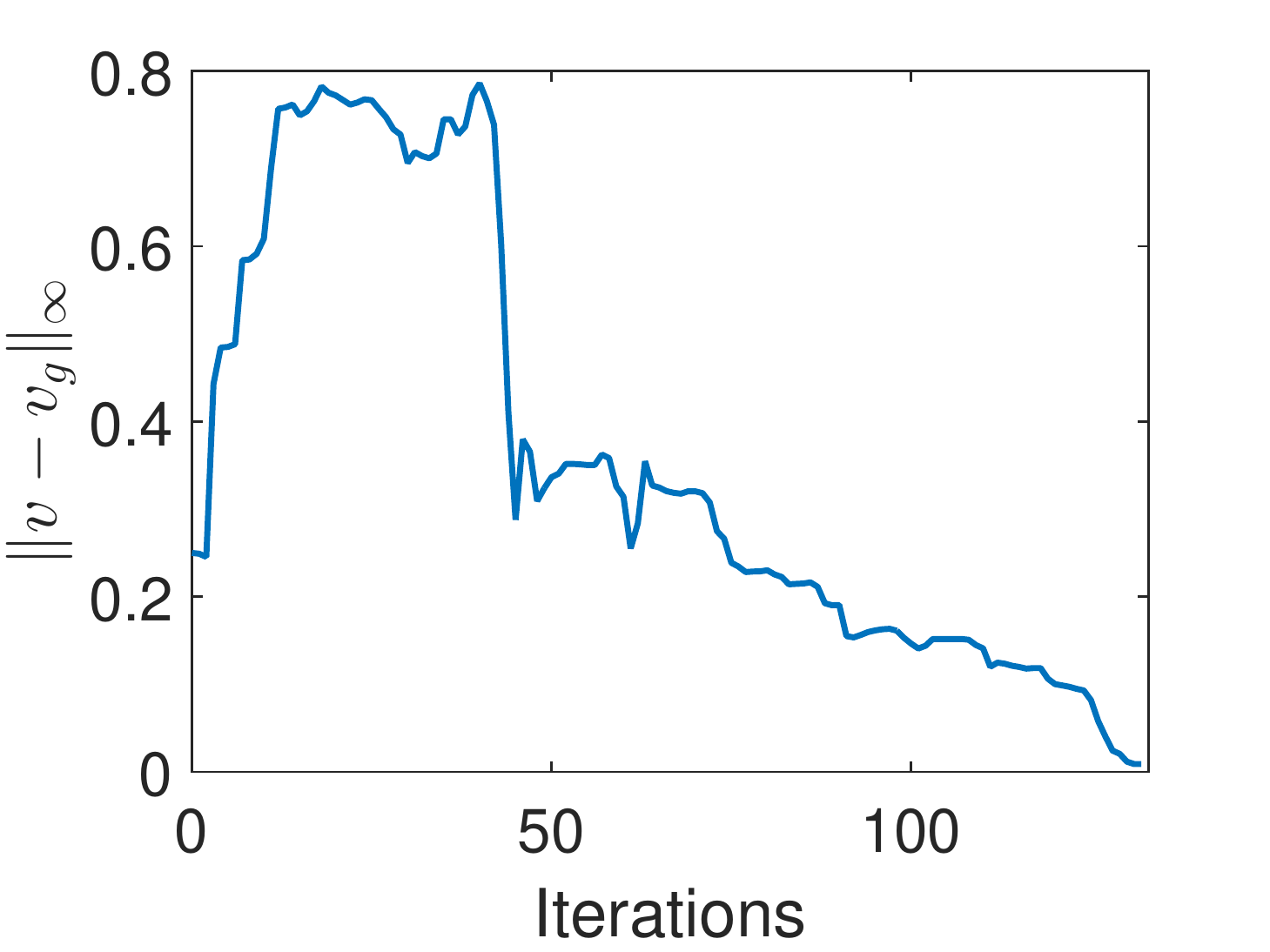}\label{fig:convergence12Param_Par}}
\subfigure[]{\includegraphics[width=0.475\textwidth]{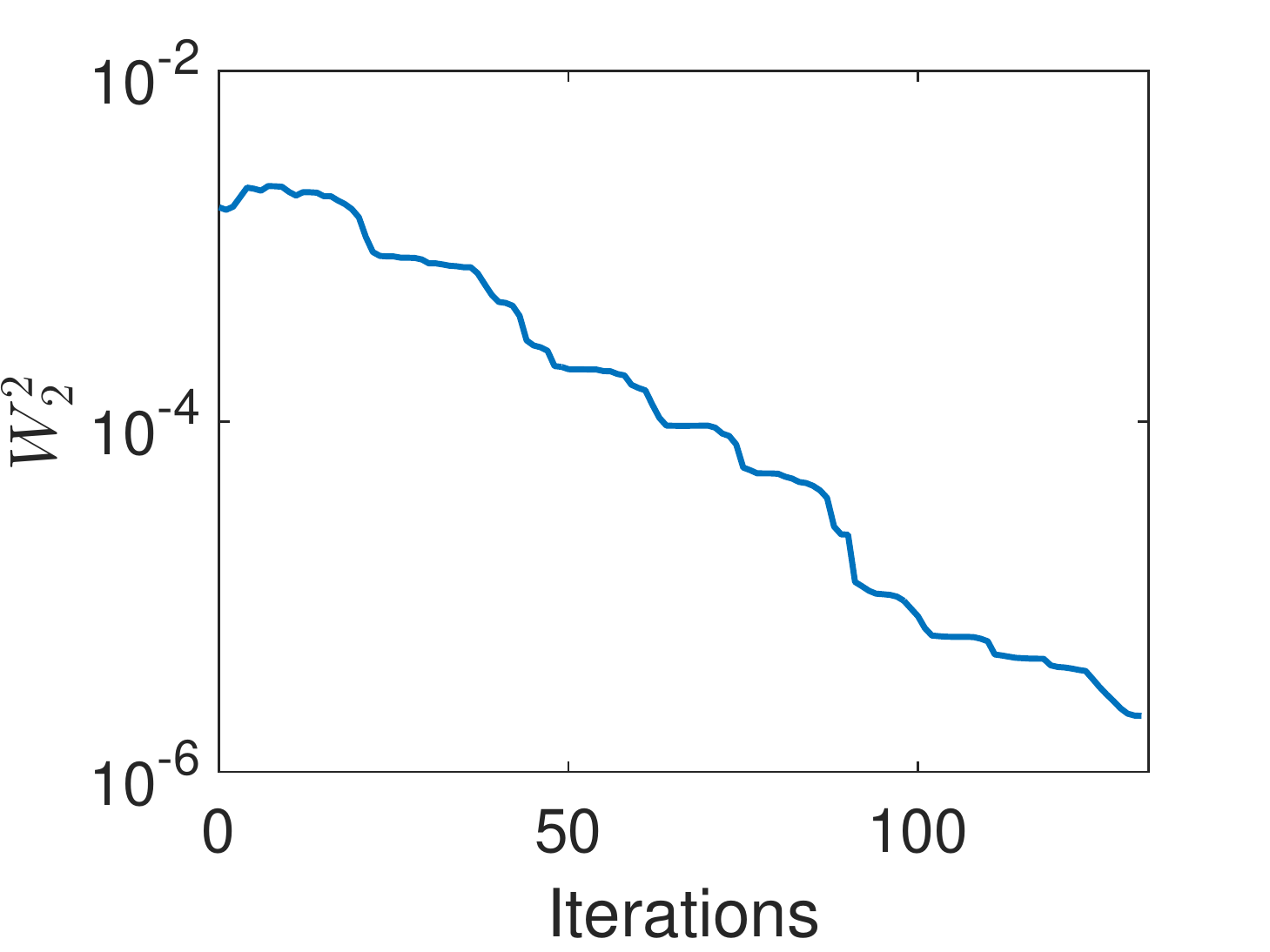}\label{fig:convergence12Param_W2}}
\caption{Convergence history for a twelve parameter model.}
\label{fig:convergence12}
\end{figure}

\section{Conclusions}\label{sec:conclusions}

In this paper, we demonstrate several advantages of the Wasserstein metric as a measure of misfit between seismic signals in connection to full waveform inversion. In particular, we proved that this distance is convex with respect to several common transformations and is less sensitive to noise than the $L_2$ distance. Additionally, the Frech\'{e}t gradient is easily computed, which makes the Wasserstein metric extremely promising for optimization and thus for seismic inversion problems. Simple numerical examples demonstrate the efficiency of using this metric.

A natural direction for future research is increasing the efficiency of the computation with quasi-Newton techniques and parallelization in order to apply the method to more realistic seismic applications.

\bibliographystyle{plain}
\bibliography{WassBib}

\end{document}